\newif\ifmhotti\mhottitrue

\documentclass[english,letterpaper,cleveref,autoref,thm-restate,colorlinks=true,linkcolor=blue,citecolor=red,nameinlink]{lipics-v2021}
\hideLIPIcs
\usepackage[LGR,OT1]{fontenc}
\usepackage[utf8]{inputenc}
\usepackage[disable]{todonotes}
\usepackage[inline]{enumitem}
\usepackage{graphicx,xspace,amsmath,wasysym}
\nolinenumbers

\newcommand{\prooflink}[1]{\hypersetup{linkcolor=magenta}\hyperref[#1]{$\triangledown$}}
\newcommand{\statlink}[1]{\hypersetup{linkcolor=blue}\hyperref[#1]{$\vartriangle$}}
\newcommand{\mylink}{}

\usepackage{subcaption}
\usepackage{dsfont}
\usepackage{booktabs}
\usepackage{xcolor}
\usepackage[nofancy]{svninfo}
\svnInfo $Id: newArcs.tex 7566 2018-10-25 15:48:32Z eth.foersth $

\usepackage{amsfonts,amssymb,bm,thm-restate}

\newcommand{\R}{\ensuremath{\mathds R}} %

\newcommand{\textgreek}[1]{\begingroup\fontencoding{LGR}\selectfont#1\endgroup}

\newcommand{\todosc}[1]{\par\noindent\begin{nolinenumbers}\todo[color=yellow!20!white,inline]{SC: #1}\end{nolinenumbers}\par\noindent}

\newcommand{\todomh}[1]{\par\noindent\begin{nolinenumbers}\todo[color=green!20!white,inline]{MH: #1}\end{nolinenumbers}\par\noindent}

\graphicspath{{figures/}}

\title{Monotone Arc Diagrams with few Biarcs} 

\date{Revision \svnInfoRevision\ --- \svnToday}
\authorrunning{Chaplick, F\"orster, Hoffmann, and Kaufmann}

\author{Steven Chaplick}{Maastricht University, The Netherlands\and \url{https://www.maastrichtuniversity.nl/sa-chaplick}}{s.chaplick@maastrichtuniversity.nl}{https://orcid.org/0000-0003-3501-4608}{supported by DFG grant WO 758/11-1}
\author{Henry F\"orster}{Universit\"{a}t T\"{u}bingen, Germany \and \url{https://uni-tuebingen.de/fakultaeten/mathematisch-naturwissenschaftliche-fakultaet/fachbereiche/informatik/lehrstuehle/algorithmik/team/dr-henry-foerster/}}{henry.foerster@uni-tuebingen.de}{https://orcid.org/0000-0002-1441-4189}{}
\author{Michael Hoffmann}{Department of Computer Science, ETH Z\"urich, Switzerland \and\url{https://people.inf.ethz.ch/hoffmann/}}{hoffmann@inf.ethz.ch}{https://orcid.org/0000-0001-5307-7106}{supported by the Swiss National Science Foundation within the collaborative D-A-CH project \emph{Arrangements and Drawings} as SNSF project 200021E-171681.}
\author{Michael Kaufmann}{Universit\"{a}t T\"{u}bingen, Germany \and \url{https://uni-tuebingen.de/fakultaeten/mathematisch-naturwissenschaftliche-fakultaet/fachbereiche/informatik/lehrstuehle/algorithmik/team/prof-dr-michael-kaufmann/}}{michael.kaufmann@uni-tuebingen.de}{https://orcid.org/0000-0001-9186-3538
}{}

\Copyright{S. Chaplick, H. F\"orster, M. Hoffmann, and M. Kaufmann}


\acknowledgements{This work started at the workshop on \emph{Graph and Network Visualization} (GNV~2017) in Heiligkreuztal, Germany.  Preliminary results were presented at the 36th European Workshop on Computational Geometry (EuroCG~2020). We thank Stefan Felsner and Stephen Kobourov for useful discussions.}  

\hideLIPIcs

\begin{document}
  

\maketitle

\begin{abstract}
  We show that every planar graph has a monotone topological $2$-page book embedding where at most $(4n-10)/5$ (of potentially $3n-6$) edges cross the spine, and every edge crosses the spine at most once; such an edge is called a \emph{biarc}. We can also guarantee that all edges that cross the spine cross it in the same direction (e.g., from bottom to top). For planar $3$-trees we can further improve the bound to~$(3n-9)/4$, and for so-called Kleetopes we obtain a bound of at most~$(n-8)/3$ edges that cross the spine. The bound for Kleetopes is tight, even if the drawing is not required to be monotone. A \emph{Kleetope} is a plane triangulation that is derived from another plane triangulation~$T$ by inserting a new vertex~$v_f$ into each face~$f$ of~$T$ and then connecting~$v_f$ to the three vertices of~$f$. 
\end{abstract} 

\section{Introduction}\label{sec:intro}

\emph{Arc diagrams} (\cref{fig:0}) are drawings of graphs that represent
vertices as points on a horizontal line, called \emph{spine}, and edges as
\emph{arcs}, consisting of a sequence of halfcircles centered on the spine. A
\emph{proper arc} consists of one halfcircle. In \emph{proper arc
  diagrams} all arcs are proper (see \cref{fig:1:c}). In \emph{plane} arc
diagrams no two edges cross. Note that proper plane arc diagrams are also known
as \emph{$2$-page book embeddings}. 
Bernhard and Kainen~\cite{bk-btg-79} characterized the graphs that admit proper plane arc
diagrams: subhamiltonian planar graphs, i.e., subgraphs of planar graphs with a
Hamiltonian cycle. In particular, non-Hamiltonian maximal planar graphs do not
admit proper plane arc diagrams.

\begin{figure}[bhtp]
  \centering%
  {\captionsetup{singlelinecheck=true}
  \begin{subfigure}[b]{.3\textwidth}
    \centering
    \includegraphics{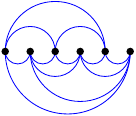}
    \subcaption{}	
    \label{fig:1:c}
  \end{subfigure} 
  \hfill
  \begin{subfigure}[b]{.3\textwidth}
    \centering
    \includegraphics{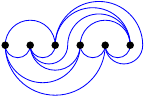} 
    \subcaption{}	
    \label{fig:1:a}  
  \end{subfigure}\hfill
  \begin{subfigure}[b]{.3\textwidth}
    \centering
    \includegraphics{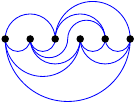}
    \subcaption{}	
    \label{fig:1:b} 
  \end{subfigure}\hfill}
  
  \caption{Arc diagrams of the octahedron: (a)~proper, (b)~general, and (c)~monotone.\label{fig:0}}
\end{figure}

To represent all planar graphs as a plane arc diagram, it suffices to allow each edge to cross the spine 
once~\cite{kw-evpfb-02}. The resulting arcs composed of two halfcircles are called \emph{biarcs} (see \cref{fig:1:a}). Additionally, all edges can be drawn as \emph{monotone} curves w.r.t.~the spine~\cite{ddlw-ccdpg-05}; such a drawing is called a \emph{monotone topological ($2$-page) book embedding} (see \cref{fig:1:b}). A monotone biarc is either \emph{up-down} or \emph{down-up}, depending on whether the left halfcircle is drawn above or below the spine, respectively. Note that a \emph{monotone topological} book embedding is not necessarily a book embedding, even though the terminology suggests it.

In general, biarcs are needed, but \emph{many} edges can be drawn as proper
arcs. Cardinal, Hoffmann, Kusters, T{\'o}th, and
Wettstein~\cite{chktw-adfdht-18} gave bounds on the required number of biarcs by
showing that every planar graph on~$n\ge 3$ vertices admits a plane arc diagram
with at most~$\lfloor(n-3)/2\rfloor$ biarcs and 
how this quantity is related to the diameter of the so-called combinatorial flip graph of triangulations. However, they allow general, not necessarily monotone biarcs. When requiring biarcs to be monotone, Di Giacomo, Didimo, Liotta, and Wismath~\cite{ddlw-ccdpg-05} gave an algorithm to construct a monotone plane arc diagram that may create close to~$2n$ biarcs for an $n$-vertex planar graph. Cardinal, Hoffmann, Kusters, T{\'o}th, and Wettstein~\cite{chktw-adfdht-18} improved this bound to at most~$n-4$ biarcs.

As a main result, we improve the upper bound on the number of monotone biarcs.

\begin{restatable}{theorem}{mainthm}\label{thm:1}
  Every $n$-vertex planar graph admits a plane arc diagram with at most
  $\left\lfloor \frac{4}{5}n \right\rfloor - 2$ biarcs that are all down-up
  monotone. 
\end{restatable}
\todosc{Where do we justify this runtime? with the recursive construction in regions it seems not obvious. I think we could claim polynomial time without needing to really explain it. I see the construction likely leading to quadratic time, but indeed there is no space to explain. }
\todomh{Right. Do we want to try to make an argument or just drop the runtime claim (for now)? I'll comment it out for now because I don't see room to argue about this in the main text.}

It is an intriguing open question if there is a \emph{monotonicity penalty}, that is, 
is there a graph~$G$ and a plane arc diagram of~$G$ with~$k$ biarcs such that every monotone plane arc diagram of~$G$ has strictly more than~$k$ biarcs? No such graph is known, even if for the stronger condition that all biarcs are 
monotone of the same type, such as down-up.

For general plane arc diagrams, in some cases $\lfloor (n-8)/3\rfloor$ biarcs are required~\cite{chktw-adfdht-18}. The (only) graphs for which this lower bound is known to be tight belong to the class of Kleetopes.
A \emph{Kleetope} is a plane triangulation\footnote{A \emph{plane} triangulation is a triangulation associated with a combinatorial embedding. For the scope of this paper, we also consider the outer face to be fixed.} that is derived from another plane triangulation~$T$ by inserting a new vertex~$v_f$ into each face~$f$ of~$T$ and then connecting~$v_f$ to the three vertices of~$f$. One might think that Kleetopes are good candidates to exhibit a monotonicity penalty. 
However, we show that this is not the case, but instead the known lower bound is tight. 

\begin{restatable}{theorem}{kleetopes}\label{thm:kleetope}
  Every Kleetope on~$n$ vertices admits a monotone plane arc diagram with at most~$\lfloor(n-8)/3\rfloor$ biarcs, where every biarc is down-up.
\end{restatable}

So, to discover a monotonicity penalty we have to look beyond Kleetopes. We investigate another class of planar graphs: planar $3$-trees. A \emph{planar~$3$-tree} is built by starting from a (combinatorial) triangle. At each step we insert a new vertex~$v$ into a (triangular) face~$f$ of the graph built so far, and connect~$v$ to the three vertices of~$f$. As a third result we give an improved upper bound on the number of monotone biarcs needed for planar $3$-trees.
\begin{restatable}{theorem}{threetrees}\label{thm:3tree}
  Every planar $3$-tree admits a plane arc diagram with at most $\left\lfloor \frac{3}{4}(n-3) \right\rfloor$ biarcs that are all down-up monotone. 
\end{restatable}
 
\subparagraph{Related work.}

Giordano, Liotta, Mchedlidze, Symvonis, and Whitesides~\cite{GiordanoLMSW15}
showed that every upward planar graph admits an \emph{upward topological book
  embedding} in which all edges are either proper arcs or biarcs. These
embeddings are also monotone arc diagrams that respect the orientations of the
edges and use at most one spine crossing per edge. One of their directions for
future work is to 
minimize the number of spine crossings. We believe that our approach for
undirected graphs may provide some insights. Everett, Lazard, Liotta, and
Wismath~\cite{DBLP:journals/dcg/EverettLLW10} used monotone arc diagrams to
construct small universal point sets for $1$-bend drawings of planar graphs,
heavily using the property that all biarcs have the same \emph{shape} (e.g., all
are down-up biarcs). This result has been extended by L\"offler and
T\'oth~\cite{LofflerT15} by restricting the set of possible bend positions. They
use the existence of monotone arc diagrams with at most $n-4$ biarcs to build
universal point sets of size $6n-10$ (vertices and bend points) for $1$-bend
drawings of planar graphs on $n$ vertices. Using \cref{thm:1}, we can 
decrease the number of points by about~$n/5$. 


\subparagraph{Outline.} 
We sketch the proof of \cref{thm:1} in \cref{sec:overview,sec:default,sec:nondefault}, then in \cref{sec:kleetopes} the proof of \cref{thm:kleetope}, and finally, in \cref{sec:3trees} the proof of \cref{thm:3tree}.
Due to space constraints, some proofs are provided in the appendix only; their statements are marked with \textcolor{blue}{$\vartriangle$}\textcolor{magenta}{$\triangledown$}. In the PDF, \textcolor{blue}{$\vartriangle$} 
links to the statement in the main text and \textcolor{magenta}{$\triangledown$}  links to the proof in the appendix.

\section{Overview of our Algorithm}\label{sec:overview}

To prove \cref{thm:1} we describe an algorithm to incrementally construct an arc
diagram for a given planar graph~$G=(V,E)$ on~$n\ge 4$~vertices. Without loss of
generality we assume that~$G$ is a combinatorial \emph{triangulation}, that is,
a maximal planar graph. Further, we consider~$G$ to be embedded, that is,
$G$~is a \emph{plane} graph. As every triangulation on~$n\ge 4$ vertices is
$3$-connected, by Whitney's Theorem selecting one facial triangle as the \emph{outer face} embeds it
into~the plane. This choice also determines a unique outer face for every
biconnected subgraph. For a biconnected plane graph~$G$ denote the outer face
(an open subset of~$\R^2$) by~$F_\circ(G)$ and denote by~$C_\circ(G)$ the cycle
that bounds~$F_\circ(G)$. A plane graph is \emph{internally triangulated} if it
is biconnected and every inner face is a triangle. A central tool for our
algorithm is the notion of a canonical ordering~\cite{fpp-hdpgg-90}.
Consider an internally triangulated plane graph~$G$ on the
vertices~$v_1,\ldots,v_n$, and let~$V_k=\{v_j\colon 1\le j\le k\}$. The
sequence~$v_1,\ldots,v_n$ forms a \emph{cano\-nical ordering} for~$G$ if the
following conditions hold for every~$i\in\{3,\ldots,n\}$:
\begin{enumerate}[label=(C\arabic*),left=\labelsep]
\item\label{co:1} the induced subgraph $G_i=G[V_i]$ is internally triangulated;
\item\label{co:2} the edge~$v_1v_2$ is an edge of~$C_\circ(G_i)$; and
\item\label{co:3} for all~$j$ with~$i<j\le n$, we have~$v_{j}\in F_\circ(G_i)$. 
\end{enumerate}
Every internally triangulated plane graph admits a canonical ordering, for any
starting pair~$v_1,v_2$ where~$v_1v_2$ is an edge
of~$C_\circ(G)$~\cite{fpp-hdpgg-90}. Moreover, such an ordering can be computed
by iteratively selecting~$v_i$, for~$i=n,\ldots,3$, to be a vertex
of~$C_\circ(G_i)\setminus\{v_1,v_2\}$ that is not incident to a chord
of~$C_\circ(G_i)$. This computation can be done in~$O(n)$
time~\cite{cp-ltadp-95}. In general, a triangulation may admit many canonical
orderings. We will use this freedom to adapt the canonical ordering we work with
to our needs.
To this end, we compute a canonical ordering for~$G$ incrementally, starting
with~$v_1,v_2,v_3$, where~$v_1v_2$ is an arbitrary edge of~$C_\circ(G)$,
and~$v_3$ is the unique vertex of~$G$ such that~$v_1v_2v_3$ bounds a triangular
face of~$G$ and~$v_3$ is not a vertex of~$C_\circ(G)$. A canonical
ordering~$v_1,\ldots,v_i$ for~$G_i$, where~$3\le i\le n$, is \emph{extensible}
if there exists a sequence~$v_{i+1},\ldots,v_n$ such that~$v_1,\ldots,v_n$ is a
canonical ordering for~$G$.

\begin{restatable}{rlemma}{lemextend}\label{lem:extend}
  \renewcommand{\mylink}{\statlink{lem:extend}\prooflink{Plemextend}}
  A canonical ordering~$v_1,\ldots,v_i$ for~$G_i$ is extensible 
  $\iff$ $V\setminus V_i\subset F_\circ(G_i)$.
\end{restatable}

\begin{figure}[htbp]
  \centering%
  \includegraphics{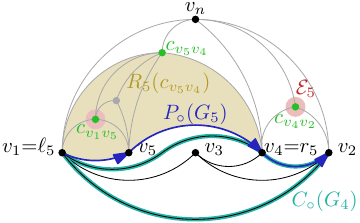}
  \caption{Overview of notation used throughout the paper.\label{fig:notation}}
\end{figure}

We set up some terminology used throughout the paper; refer also to \cref{fig:notation}.
Consider an extensible canonical ordering~$v_1,\ldots,v_i$ for~$G_i$ and some vertex~$v\in V\setminus V_i$. Let~$P_\circ(G_i)$ denote the path~$C_\circ(G_i)-v_1v_2$ and direct it from~$v_1$ to~$v_2$. 
As~$G_i$ is an induced subgraph of the plane graph~$G$ and~$v\in F_\circ(G_i)$ (by extensibility), all neighbors of~$v$ in~$G_i$ are on~$P_\circ(G_i)$. We associate a planar region~$R_i(v)$ to~$v$ as follows. If~$d_i(v)=\deg_{G_i}(v)\le 1$, then~$R_i(v)=F_\circ(G_i)$; else, let~$R_i(v)$ be the open bounded region bounded by the simple closed curve formed by the part of~$P_\circ(G_i)$ between~$\ell$ and~$r$ together with the edges~$\ell v$ and~$rv$ of~$G$,
where~$\ell$ and~$r$ are the first and last, respectively, neighbor (in~$G$)
of~$v$ on~$P_\circ(G_i)$. We partially order the vertices in~$V\setminus V_i$ by
defining~$v\prec v'$ if~$R_i(v)\subseteq R_i(v')$. 

A vertex~$v\in V\setminus V_i$ is \emph{eligible} (for~$G_i$) if setting~$v_{i+1}=v$ yields an extensible canonical ordering~$v_1,\ldots,v_{i+1}$ for~$G_{i+1}$. Denote the set of vertices eligible for~$G_i$ by~$\mathcal{E}_i$. 
Let~$e=uw$ be an arbitrary edge of~$P_\circ(G_i)$, for~$i<n$. As~$G$ is a
triangulation, there exists a unique vertex~$c_e\in V\setminus V_i$ such
that~$uwc_e$ bounds a triangular face of~$G$; we say that~$c_e$
\emph{covers}~$e$. 
Given a canonical ordering $v_1,\ldots,v_n$, vertex $v_i$ covers exactly the edges of $P_\circ(G_{i-1})$ that are not on~$P_\circ(G_{i})$. Similarly, we say that $v_i$ covers a vertex $v$ of $P_\circ(G_{i-1})$ if $v$ is not part of $P_\circ(G_{i})$. 
The following observations are direct consequences of these definitions and \cref{lem:extend}.

\begin{corollary}\label{lem:eligible}
  A vertex~$v\in V\setminus V_i$ is eligible $\iff R_i(v)\cap V=\emptyset\iff R_i(v)\cap\mathcal{E}_i=\emptyset$.
\end{corollary}

While computing a canonical ordering~$v_1,\ldots,v_n$, we also maintain an arc diagram, for short, \emph{diagram} of~$G_i$. This diagram must satisfy certain properties to be considered valid, as detailed below. In some cases we apply
induction to handle a whole induced subgraph of~$G$, for instance, within a (separating) triangle, at once. As a result, in certain steps, subgraph $G_i$ may not correspond to a valid arc diagram. 

Every vertex~$v_i$ arrives with $1-\chi$~credits, for some constant~$\chi\ge 0$.\footnote{For \cref{thm:1} we will set~$\chi=1/5$. But we think it is instructive to keep~$\chi$ as a general constant in our argumentation. For instance, this way it is easier to see in which cases our analysis is tight.}
For these credits we can either create biarcs (at a cost of one
credit per biarc), or we place them on edges of the outer face of the diagram
for later use. The \emph{costs}~$\mathrm{cost}(D)$ of a diagram~$D$ is the sum
of credits on its edges.
An edge in the diagram can be one of three different types: \emph{mountain}
(proper arc above the spine), \emph{pocket} (proper arc below the spine), or
\emph{down-up biarc}. So the diagram is determined by (1)~the spine order
(left-to-right) of the vertices and crossings along with (2)~for every edge, its
type and number of credits. The \emph{lower envelope} of a diagram consists of
all vertices and edges that are vertically visible from below, that is, there is
no other vertex or edge of the diagram vertically below it. Analogously, the
\emph{upper envelope} consists of all vertices and edges that are vertically
visible from above.

A 
diagram for~$v_1,\ldots,v_i$ and~$i\in\{3,\ldots,n\}$, is \emph{valid}
if it satisfies the following invariants: 
\begin{enumerate}[label={(I\arabic*)},series=inv]\setlength{\itemindent}{\labelsep}
\item\label{i:biarcTypes} Every edge is either a proper arc or a down-up biarc. Every edge on the upper envelope is a proper arc.
\item\label{i:mountainMoney} Every mountain whose left endpoint is
  on~$C_\circ(G_i)\setminus\{v_2\}$ carries one credit.
\item\label{i:biarcMoney} Every biarc carries (that is, is paid for with) one
  credit.
\item\label{i:pocketMoney} Every pocket on~$P_\circ(G_i)$
  carries~$\chi$~credits\footnote{
  As in the Greek word for pocket money: \textgreek{χαρτζιλίκι}.}.
\end{enumerate}

Moreover, a valid drawing is \emph{extensible} if it also satisfies
\begin{enumerate}[label={(I\arabic*)},series=inv]\setlength{\itemindent}{\labelsep}
  \setcounter{enumi}{4}
\item\label{i:contour} Vertex~$v_1$ is the leftmost and $v_2$ is the rightmost
  vertex on the spine. Edge~$v_1v_2$ forms the lower envelope of
  $C_\circ(G_i)$. The edges of~$P_\circ(G_i)$ form the upper envelope.
  \end{enumerate}


\noindent To prove \cref{thm:1} it suffices to prove the following.

\begin{lemma}\label{lem:main}
  Let~$G$ be a maximal plane graph on~$n\ge 3$ vertices, let~$v_1,\ldots,v_i$ be
  an extensible canonical ordering for~$G_i$, for some~$3\le i<n$, and let~$D$ be
  an extensible arc diagram for~$G_i$. Then, for any $\chi \leq \frac{1}{5}$, $D$ can be extended to an extensible
  arc diagram~$D'$ for~$G$
  with~$\mathrm{cost}(D')\le\mathrm{cost}(D)+(n-i)(1-\chi)+\xi$, for
  some~$\xi\le 2\chi$.
  %
\end{lemma}

\begin{proof}[Proof of \cref{thm:1} assuming \cref{lem:main}]
  We may assume~$n\ge 4$, as the statement is trivial for~$n\le
  3$. Let~$C_\circ(G)=v_1v_2v_n$, and let~$v_3$ be the other (than~$v_n$) vertex
  that forms a triangle with~$v_1v_2$ in~$G$. Then~$v_1,v_2,v_3$ is an
  extensible canonical ordering for~$G_3$ in~$G$. To obtain an extensible 
  diagram~$D$ for~$G_3$, place~$v_1v_3v_2$ on the spine in this order from left
  to right. All three edges are drawn as pockets so that~$v_1v_2$ is
  below~$v_1v_3$ and~$v_3v_2$. On the latter two edges we put~$\chi$~credits
  each. It is easily verified that~$D$ is extensible
  and~$\mathrm{cost}(D)=2\chi$. By \cref{lem:main} we obtain an extensible
  diagram~$D'$ for~$G$
  with~$\mathrm{cost}(D')\le 2\chi+(n-3)(1-\chi)+2\chi=n(1-\chi)+7\chi-3$. 
  Setting $\chi=1/5$ yields $\mathrm{cost}(D')\le \frac{4}{5}n-\frac{8}{5}$.
  As $v_n$ is incident to a mountain on the outer face by \ref{i:contour} which carries a credit by \ref{i:mountainMoney}, $\mathrm{cost}(D')-1$ is an 
  upper bound for the number of biarcs in~$D'$ and the theorem follows. 
\end{proof}


We can avoid the additive term~$\xi$ in \cref{lem:main} by dropping \ref{i:contour} for $D'$:

\begin{lemma}\label{lem:mainadapt}
  Let~$G$ be a maximal plane graph on~$n\ge 4$ vertices, let~$v_1,\ldots,v_i$ be
  an extensible canonical ordering for~$G_i$, for~$3\le i<n$, and let~$D$ be an
  extensible arc diagram for~$G_i$. Then, for any $\chi \leq \frac{1}{5}$, $D$ can be extended to a valid arc
  diagram~$D'$ for~$G$ such that
  (1)~$\mathrm{cost}(D')\le\mathrm{cost}(D)+(n-i)(1-\chi)$, (2)~Vertex~$v_1$ is
  the leftmost and~$v_n$ is the rightmost vertex on the spine. The
  mountain~$v_1v_n$ forms the upper envelope, and the pocket~$v_1v_2$ along with
  edge $v_2v_n$ forms the lower envelope of~$D'$, and
  (3)~$v_2v_n$ is not a pocket. 
\end{lemma}


\section{Default vertex insertion}\label{sec:default}

We prove both \cref{lem:main} and \cref{lem:mainadapt} together by induction
on~$n$. For \cref{lem:main}, the base case~$n=3$ is trivial, with~$D'=D$. For
\cref{lem:mainadapt}, the base case is~$n=4$ and~$i=3$. We place~$v_4$ as
required, to the right of~$v_2$, and draw all edges incident to~$v_4$ as
mountains. To establish \ref{i:mountainMoney} it suffices to put one credit
on~$v_1v_4$ because~$v_3$ is covered by~$v_4$ and mountains with left
endpoint~$v_2$ are excluded in \ref{i:mountainMoney}. The edge of~$D$ with left
endpoint~$v_3$ is covered by~$v_4$; thus, we can take the at least~$\chi$
credits on it. The invariants \ref{i:biarcTypes}, \ref{i:biarcMoney}, and
\ref{i:pocketMoney} are easily checked to hold, as well as the statements in
\cref{lem:mainadapt}.

In order to describe a generic step of our
algorithm, assume that we already have an extensible arc diagram for~$G_{i-1}$,
for~$i=4,\ldots,n$. We have to select an eligible
vertex~$V_i\in V\setminus V_{i-1}$ and add it using at most~$1-\chi$ credits obtaining an extensible diagram for~$G_i$. In this section we discuss
some cases where a suitable vertex exists that can easily be added to the arc
diagram, using what we call a \emph{default insertion}. Let~$v_i$ be any vertex
in~$\mathcal{E}_{i-1}$.


We call the sequence
of (at least one) edges of~$P_\circ(G_{i-1})$ between the leftmost neighbor $\ell_i$ of $v_i$ and the rightmost neighbor $r_i$ of $v_i$ the
\emph{profile}~$\mathrm{pr}(v_i)$ of~$v_i$. By~\ref{i:biarcTypes} each edge on the
profile is a pocket or a mountain, i.e., writing~$\smile$ and~$\frown$
for pocket and mountain, respectively, each profile can be described by a string
over~$\{\smile,\frown\}$. 
For a set~$A$ of characters, let~$A^*$, $A^k$ and~$A^+$ denote the set of all strings, all strings of length exactly $k$  and all strings of length at least one, respectively, formed by characters from~$A$.
Let~$d_i$ denote the degree of~$v_i$ in~$G_i$.

\begin{restatable}{rlemma}{lemdefaultApproachone}\label{lem:defaultApproach1}
  \renewcommand{\mylink}{\statlink{lem:defaultApproach1}\prooflink{PlemdefaultApproach1}}
  If~$\mathrm{pr}(v_i)\in\{\smile,\frown\}^*\smile\frown^*$, then we can
  insert~$v_i$ and use~$\le 1$ 
  credit to obtain an extensible arc diagram for~$G_i$. At most~$1-\chi$ credits suffice, unless~$\mathrm{pr}(v_i)=\;\frown\smile$.
\end{restatable}
\begin{proof}[Proof Sketch]
  We place~$v_i$ into the rightmost pocket~$p_{\ell}p_r$ it covers, draw~$p_{\ell}v_i$ and~$v_ip_r$ as pockets and all other new edges as mountains; see \cref{fig:naive2}. We take the~$\chi$ credits from~$p_{\ell}p_r$. If~$d_i=2$, then we place~$\chi$ credits on each of the two pockets incident to~$v_i$ so as to establish~\ref{i:pocketMoney}, for a cost of~$\chi\le 1-\chi$, assuming~$\chi\le 1/2$. 
  
  \begin{figure}[htbp]
    \centering%
    \begin{minipage}[b]{.45\textwidth}
      \centering
      \includegraphics[page=3]{arcDiagramsFigures}
    \end{minipage}\hfill
    \begin{minipage}[b]{.45\textwidth}
      \centering
      \includegraphics[page=1]{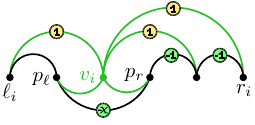}
    \end{minipage}
    \caption{Inserting a vertex $v_i$ into a pocket, using $1-\chi$ credits
      (\cref{lem:defaultApproach1}).\label{fig:naive2}}
  \end{figure}

  For~$d_i\ge 3$ 
  each new mountain~$m$ from~$v_i$ to the right covers a mountain~$m'$ of~$P_\circ(G_{i-1})$ whose left endpoint is covered by~$v_i$, Thus, we can take the credit from~$m'$ and place it on~$m$. Among all mountains from~$v_i$ to the left, a credit is needed for the leftmost one only. If there is such a mountain, then we do not need the~$\chi$ credits on~$p_\ell v_i$. And if~$v_i$ covers two or more edges to the left of~$p_\ell$, we gain at least~$\chi$ credits from the rightmost such edge.
\end{proof}

It is more difficult to insert~$v_i$ if it covers mountains only, at least
if~$d_i$ is small. But if the degree of~$v_i$ is large, then we can actually
gain credits by inserting~$v_i$ (see \cref{fig:mountains2}).

\begin{restatable}{rlemma}{lemdefaultApproachtwo}\label{lem:defaultApproach2}
  \renewcommand{\mylink}{\statlink{lem:defaultApproach2}\prooflink{PlemdefaultApproach2}}
  If~$\mathrm{pr}(v_i)\in\;\frown^+$ and~$d_i\ge 5$, then we can insert~$v_i$
  and gain at least~$d_i-5$ credits to obtain an extensible arc diagram for~$G_i$.
\end{restatable}
 
\begin{figure}[htbp]
  \hfill%
  \includegraphics[page=8]{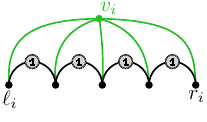}\hfill
  \includegraphics[page=7]{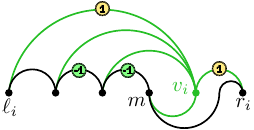}\hfill\hfill
  \caption{Inserting a vertex~$v_i$ into mountains, using $5-d_i$ credits
    (\cref{lem:defaultApproach2}).\label{fig:mountains2}}
\end{figure}

An eligible vertex is \emph{problematic} if it is of one of the four specific
types depicted in~\cref{fig:open}. Using Lemmas~\ref{lem:defaultApproach1}
and~\ref{lem:defaultApproach2} we 
insert vertices using at most~$1-\chi$ credits per vertex, unless all eligible vertices are problematic. This specific situation is discussed in the next section.

\begin{figure}[htbp]
  \centering
  {\captionsetup{singlelinecheck=true}
  \begin{minipage}[b]{.2\textwidth}
    \centering
    \includegraphics[page=1]{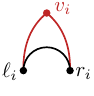}
    \subcaption{$\mathcal{T}(2,\frown)$}  
    \label{fig:open:2}  
  \end{minipage}\hfill
  \begin{minipage}[b]{.24\textwidth}
    \centering
    \renewcommand{\thelinenumber}{}%
    \includegraphics[page=2]{open}
    \subcaption{$\mathcal{T}(3,\frown^2)$}  
    \label{fig:open:3}  
  \end{minipage}\hfill
  \begin{minipage}[b]{.28\textwidth}
    \centering
    \renewcommand{\thelinenumber}{}%
    \includegraphics[page=3]{open}
    \subcaption{$\mathcal{T}(4,\frown^3)$}  
    \label{fig:open:4}  
  \end{minipage}\hfill
  \begin{minipage}[b]{.24\textwidth}
    \centering
    \renewcommand{\thelinenumber}{}%
    \includegraphics[page=4]{open}
    \subcaption{$\mathcal{T}(3,\frown\smile)$}  
    \label{fig:open:1}  
  \end{minipage}\hfill}
  \caption{The four types of problematic vertices where default insertion
    fails.\label{fig:open}}
\end{figure}

\section{When default insertion fails}
\label{sec:nondefault}

In this section we discuss how to handle the case where all eligible vertices are problematic, that is, they cannot be handled by our default insertion. Let~$v$ be an arbitrary vertex in~$\mathcal{E}_{i-1}$, and let~$\ell$ and~$r$ denote the leftmost and rightmost neighbor of~$v$ on~$P_\circ(G_{i-1})$, respectively.

A special case arises if~$v=v_n$ is the last vertex of the canonical ordering. This case is easy to resolve, see \cref{sec:lemiisn} for details. Otherwise, we have~$i<n$ and pick a \emph{pivot vertex}~$\mathrm{p}(v)$ as
follows: If~$v$ is~$\mathcal{T}(3,\frown\smile)$ we set~$\mathrm{p}(v)=r$ and
say that~$v$ has \emph{right pivot type}, in the three remaining cases we
set~$\mathrm{p}(v)=\ell$ and say that~$v$ has \emph{left pivot type}.
Let~$\mathrm{pc}(v)\in V\setminus V_i$ denote the unique vertex that covers the
\emph{pivot edge}~$v\mathrm{p}(v)$.

\begin{lemma}\label{lem:degreetwo}
  Assume there is a vertex~$v\in\mathcal{E}_{i-1}$ such that~$\mathrm{pc}(v)$
  has only one neighbor on~$P_\circ(G_{i-1})$. Then we can set~$v_i=v$
  and~$v_{i+1}=\mathrm{p}(v)$ and spend at most~$1+2\chi$ credits to obtain an
  extensible arc diagram for~$G_{i+1}$.
\end{lemma}
\begin{proof}
  The resulting diagram is shown in~\cref{fig:degreetwo}. The costs to establish
  are~$1+\chi$ for~$\mathcal{T}(3,\frown\smile)$ and~$1+2\chi$ for the other types. Note that~$1+2\chi\le 2(1-\chi)$, for~$\chi\le 1/4$.
\end{proof}


\begin{figure}[htbp]
  \centering
  {\captionsetup{singlelinecheck=true}
  \begin{minipage}[b]{.48\textwidth}
    \centering
    \includegraphics[page=5]{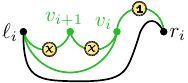}
    \subcaption{$\mathcal{T}(2,\frown)$}  
    \label{fig:degreetwo:1}  
  \end{minipage}\hfill
  \begin{minipage}[b]{.48\textwidth}
    \centering
    \renewcommand{\thelinenumber}{}%
    \includegraphics[page=6]{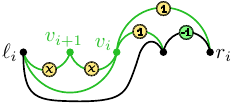}
    \subcaption{$\mathcal{T}(3,\frown^2)$}  
    \label{fig:degreetwo:2}  
  \end{minipage}\hfill\\[\baselineskip]
  \begin{minipage}[b]{.48\textwidth}
    \centering
    \includegraphics[page=7]{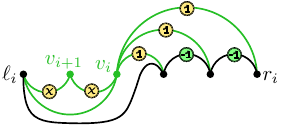}
    \subcaption{$\mathcal{T}(4,\frown^3)$}  
    \label{fig:degreetwo:3}  
  \end{minipage}\hfill
  \begin{minipage}[b]{.48\textwidth}
    \centering
    \renewcommand{\thelinenumber}{}%
    \includegraphics[page=8]{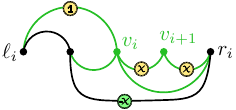}
    \subcaption{$\mathcal{T}(3,\frown\smile)$}  
    \label{fig:degreetwo:4}  
  \end{minipage}\hfill}
  \caption{Insertion of~$v_i$ and~$v_{i+1}$ if~$v_{i+1}=\mathrm{pc}(v_i)$ has
    degree two in~$G_{i+1}$.\label{fig:degreetwo}}
\end{figure}

\begin{lemma}\label{lem:samepivot}
  Assume that there are~$v,v'\in\mathcal{E}_{i-1}$ such that~$\mathrm{pc}(v)=v'$
  and at least one of~$v,v'$ has right pivot type. Then we can set~$v_i=v$
  and~$v_{i+1}=v'$ and spend at most one credit to obtain an extensible arc
  diagram for~$G_{i+1}$.
\end{lemma}
\begin{proof}
  If both~$v$ and~$v'$ have right pivot type, then we use the diagram shown
  in~\cref{fig:rightpivot}~(left). The costs are~$1-\chi\le 2(1-\chi)$, for~$\chi\le 1$.
  Otherwise, one of~$v,v'$ has left pivot type and the other has right pivot
  type, then~$\mathrm{p}(v)=\mathrm{p}(v')$ and~$\mathrm{pc}(v')=v$. As the
  roles of~$v$ and~$v'$ are symmetric, we may assume w.l.o.g. that~$v$ has right pivot type and~$v'$ has left pivot type. We use
  the diagram shown in~\cref{fig:rightpivot}~(right) for the case where~$v'$
  is~$\mathcal{T}(3,\frown^2)$; other types are handled analogously. The
  costs to establish the invariants are~$1\le 2(1-\chi)$, for~$\chi\le 1/2$.
\end{proof}

\begin{figure}[htbp]
  \begin{minipage}[b]{.49\textwidth}
    \centering\includegraphics[page=15]{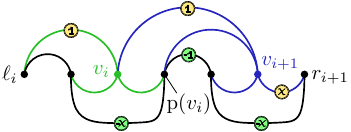}
    \label{fig:rightpivot1}
  \end{minipage}\hfill
  \begin{minipage}[b]{.49\textwidth}
    \centering\includegraphics[page=9]{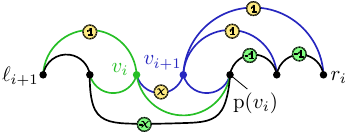}
    \label{fig:rightpivot2}  
  \end{minipage}
  \caption{Insertion of~$v_i$ and~$v_{i+1}=\mathrm{pc}(v_i)\in\mathcal{E}_{i-1}$
    if~$v_i$ has right pivot type.\label{fig:rightpivot}}
\end{figure}

If we can apply one of \cref{lem:degreetwo,lem:samepivot}, we make progress by inserting two vertices $v_i$ and $v_{i+1}$. Hence, from now on, we assume that neither of \cref{lem:degreetwo,lem:samepivot} can be applied. Our goal in the remainder of this section is to show that in this case we can find a vertex~$u$ that is not eligible but sufficiently close to being eligible---in a way described in the following---that we can aim to insert~$u$ next, along with some other vertices. 

More specifically, the vertex~$u$ has neighbors $w_1,\ldots,w_k$ on~$P_\circ(G_{i-1})$, for~$k\ge 2$, and each subregion $X_j$ of~$R_{i-1}(u)$ bounded by the edges~$uw_j$ and~$uw_{j+1}$ has a particularly simple structure. First of all, there exists an integer~$s=s(X_j)$ such that
we have~$X_j\cap\mathcal{E}_{i-1}=\{c_1,\ldots,c_s\}$, and every~$c_\ell$, for~$1\le\ell\le s$, is adjacent to~$u$ in~$G$.
We distinguish three types of regions, depending on whether~$X_j$ contains eligible vertices of left, right, or both pivot types.

\subparagraph*{Left-pivot region.}
(see \cref{fig:overview:1})
\begin{itemize}
\item Every~$c_\ell$, for $1\le\ell\le s$, has left pivot type. 
\item We have $\mathrm{pc}(c_1)=u$ and $\mathrm{pc}(c_\ell)=c_{\ell-1}$, for all~$2\le\ell\le s$.
\item All vertices in~$(V\setminus\mathcal{E}_{i-1})\cap X_j$ lie inside the face bounded by~$uc_sw_{j+1}$.
\end{itemize}

\begin{figure}[htbp]
  \centering
  {\captionsetup{singlelinecheck=true}
  \begin{subfigure}[b]{0.28\textwidth}
    \centering
    \includegraphics[scale=.96,page=46]{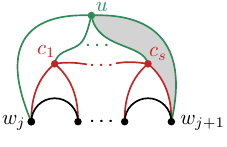}
    \subcaption{}
    \label{fig:overview:1}
  \end{subfigure}\hfill
  \begin{subfigure}[b]{0.23\textwidth}
    \centering
    \renewcommand{\thelinenumber}{}%
    \includegraphics[scale=.96,page=47]{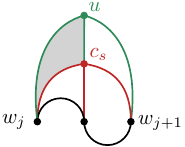}
    \subcaption{}
    \label{fig:overview:2}
  \end{subfigure}\hfill
  \begin{subfigure}[b]{0.49\textwidth}
    \centering
    \renewcommand{\thelinenumber}{}%
    \includegraphics[scale=.96,page=45]{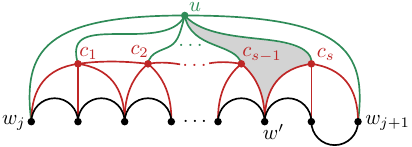}
    \subcaption{}
    \label{fig:overview:3}
  \end{subfigure}\hfill}
  \caption{Structure of regions that our to-be-inserted-next vertex~$u$ spans with~$P_\circ(G_{i-1})$. All eligible vertices (shown red) are adjacent to~$u$, all other vertices lie inside the shaded region. }
  \label{fig:overview}
\end{figure}

\subparagraph*{Right-pivot region.}
(see \cref{fig:overview:2})
\begin{itemize}
\item We have~$s=1$, the vertex~$c_1$ has right pivot type, and $\mathrm{pc}(c_1)=u$.
\item All vertices in~$(V\setminus\mathcal{E}_{i-1})\cap X_j$ lie inside the face bounded by~$uw_jc_1$. 
\end{itemize}

\subparagraph*{Both-pivot region.}
(see \cref{fig:overview:3})
\begin{itemize}
\item Every~$c_\ell$, for $1\le\ell\le s-1$, has left pivot type and~$c_s$ has right pivot type.
\item We have~$\mathrm{pc}(c_1)=\mathrm{pc}(c_s)=u$ and~$\mathrm{pc}(c_\ell)=c_{\ell-1}$, for all~$2\le\ell\le s-1$.
\item The rightmost neighbor of~$c_{s-1}$ on~$P_{\circ}(G_{i-1})$ is the same as the leftmost neighbor of~$c_s$ on~$P_{\circ}(G_{i-1})$; denote this vertex by~$w'$.
\item All vertices in~$(V\setminus\mathcal{E}_{i-1})\cap X_j$ lie inside the quadrilateral~$uc_{s-1}w'c_s$. 
\end{itemize}

\subparagraph*{How to select u.}
In the remainder of this section we will sketch how to select a suitable vertex~$u$ such that all regions spanned by~$u$ and~$P_\circ(G_{i-1})$ have the nice structure explained above.
The first part of the story is easy to tell: We select~$u$ to be a minimal (w.r.t.~$\prec$) element of the set~$\mathcal{U}:=\{\mathrm{pc}(v)\colon v\in\mathcal{E}_{i-1}\}\setminus\mathcal{E}_{i-1}$. 
Such a vertex always exists because 

\begin{restatable}{rlemma}{lemu}\label{lem:u}
  \renewcommand{\mylink}{\statlink{lem:u}\prooflink{Plemu}}
  We have~$\mathcal{U}\ne\emptyset$.
\end{restatable}

As there is a vertex~$v\in\mathcal{E}_{i-1}$ with~$u=\mathrm{pc}(v)$, we know
that~$u\in \mathcal{U}$ has at least one neighbor on~$P_\circ(G_{i-1})$, which
is~$\mathrm{p}(v)$. By \cref{lem:degreetwo} we may assume~$d_{i-1}(u)\ge
2$. Let~$w_1,\ldots,w_k$ denote the sequence of neighbors of~$u$
along~$P_\circ(G_{i-1})$. The edges~$uw_j$, for~$2\le j\le k-1$,
split~$R_{i-1}(u)$ into~$k-1$ subregions; let~$X_j$ denote the (open) region
bounded by~$w_juw_{j+1}$ and the part of~$P_\circ(G_{i-1})$ between~$w_j$
and~$w_{j+1}$, for~$1\le j<k$.

\begin{restatable}{rlemma}{lemonepivot}\label{lem:onepivot}
 \renewcommand{\mylink}{\statlink{lem:onepivot}\prooflink{Plemonepivot}}
  In every region~$X_j$, for~$1\le j<k$, there is at most one eligible
  vertex~$v$ of each pivot type for which~$\mathrm{pc}(v)=u$.
\end{restatable}


\begin{restatable}{rlemma}{lemrightpivot}\label{lem:rightpivot}
 \renewcommand{\mylink}{\statlink{lem:rightpivot}\prooflink{Plemrightpivot}}
  In every region~$X_j$, at most one eligible vertex has right pivot type. If
  there exists a vertex~$v\in X_j\cap\mathcal{E}_{i-1}$ that has right pivot
  type, then~$\mathrm{pc}(v)=u$.
\end{restatable}

\begin{restatable}{rlemma}{lemleftpivot}\label{lem:leftpivot}
 \renewcommand{\mylink}{\statlink{lem:leftpivot}\prooflink{Plemleftpivot}}
  Let~$Q$ denote the set of vertices in~$X_j\cap\mathcal{E}_{i-1}$ that have
  left pivot type. If~$Q\ne\emptyset$, then the vertices in~$Q$ form a
  sequence~$x_1,\ldots,x_q$, for some~$q\ge 0$, such that
  $x_j=\mathrm{pc}(x_{j+1})$, for~$1\le j\le q-1$, and~$\mathrm{pc}(x_1)=u$.
\end{restatable}


\begin{restatable}{rlemma}{lemregions}\label{lem:regions}
 \renewcommand{\mylink}{\statlink{lem:regions}\prooflink{Plemregions}}
  Let~$e\in P_\circ(G_{i-1})\cap\partial X_j$, for some~$1\le j<k$, and
  let~$c_e\in V\setminus V_{i-1}$ denote the vertex that covers~$e$. Then either
  $c_e=u$ or~$c_e\in\mathcal{E}_{i-1}$.
\end{restatable}

We process the regions~$X_1,\ldots,X_{k-1}$ together
with~$u$. Consider region~$X_j$ such that~$X_j\cap V\ne\emptyset$, and
denote~$E_j=P_\circ(G_{i-1})\cap\partial X_j$.  By \cref{lem:regions} the
vertices that cover one or more edges of~$E_j$ are exactly the vertices
in~$\mathcal{E}_{i-1}\cap X_j$. Thus, we can order these vertices from left to
right, according to the edge(s) in~$E_j$ they cover.  Denote this sequence
by~$c_1,\ldots,c_s$. By \cref{lem:rightpivot} the only vertex in~$X_j\cap V$ that may have right pivot type is~$c_s$. Denote~$s'=s-1$ if~$c_s$ has
right pivot type, and~$s'=s$, otherwise; i.e., $c_{s'}$ is the rightmost vertex of the sequence that has left pivot type. By \cref{lem:leftpivot} we
have~$c_h=\mathrm{pc}(c_{h+1})$, for~$1\le h \le s'-1$,
and~$\mathrm{pc}(c_1)=u$. It follows that the rightmost vertex~$w'$
of~$P_\circ(G_{i-1})$ that is adjacent to~$c_{s'}$ is the only vertex
of~$P_\circ(G_{i-1})$ that can be adjacent to a vertex
in~$(X_j\cap V)\setminus\mathcal{E}_{i-1}$. So the general situation
inside~$X_j$ can be summarized as depicted in \cref{fig:regionsummary}. Neither
the sequence of left pivot vertices nor the right pivot vertex may exist, but if
neither is present, then~$X_j\cap V=\emptyset$.

\begin{figure}[htbp]
  \centering\includegraphics[page=19]{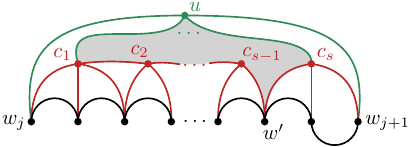}
  \caption{The structure of eligible vertices within a region~$X_j$. All
    triangular faces here are empty, only the central face (shaded) may contain other
    vertices or edges~$uc_h$, for~$2\le h<s$. The left pivot vertices could be of any
    type~$\mathcal{T}(z,\frown^{z-1})$.\label{fig:regionsummary}}
\end{figure}

The following lemma allows us to assume that the central face in each region~$X_j$ is subdivided into empty (of vertices) triangles and at most one---not necessarily empty--- triangle or quadrilateral (the latter if~$X_j$ contains eligible vertices of both pivot types).

\begin{restatable}{lemma}{lemmanyleft}
    \label{lem:manyleft}
  Let $X_j$ be a region s.t. there exist~$v,v'\in\mathcal{E}_{i-1}\cap X_j$ with~$\mathrm{pc}(v)=v'$, let~$v''$ be the vertex that covers $vv'$. If~$v''\ne u$ and~$\chi\le 1/5$, there exist~$v_i,\ldots,v_{i+h-1}$ with~$h\ge 3$ 
  s.t.~a valid 
  diagram for~$G_{i+h-1}$ can be obtained by spending at most~$(1-\chi)h$ credits.
\end{restatable}
\begin{proof}
  By \cref{lem:rightpivot} both~$v$ and~$v'$ have left pivot type. In particular, if $c_s \neq c_{s'}$, this implies that we have $v,v' \neq c_s$ (see also \cref{fig:regionsummary}). By planarity
  and as~$v''\ne u$, we have~$v''\in X_j$. If~$v''$ is not adjacent to~$w'$,
  then~$v''$ is eligible after adding~$v$ and~$v'$ and we can set~$v_i=v$,
  $v_{i+1}=v'$, and~$v_{i+2}=v''$ and use the diagram for~$G_{i+2}$ shown
  in~\cref{fig:twoleft}~(left), for a cost of~$2+2\chi\le 3-3\chi$,
  for~$\chi\le 1/5$. The figure shows the drawing where both~$v$ and~$v'$
  are~$\mathcal{T}(2,\frown)$; it easily extends to the
  types~$\mathcal{T}(3,\frown^2)$ and~$\mathcal{T}(4,\frown^3)$ because more
  mountains to the right of~$v$ can be paid for by the corresponding mountains
  whose left endpoint is covered by~$v$ and for more mountains to the left
  of~$v'$ their left endpoint is covered by~$v'$.

  Otherwise, $v''$ is adjacent to~$w'$. We claim that we may assume~$v=c_{s'}$
  and~$v'=c_{s'-1}$. To see this let~$\tilde{v}\ne v''$ be the vertex that
  covers~$c_{s'-1}c_{s'}$ and observe that~$\tilde{v}$ is enclosed by a cycle
  formed by~$vv''w'$ and the part of~$P_\circ(G_{i-1})$ between the right
  neighbor of~$v$ and~$w'$. In particular, we have~$\tilde{v}\ne u$ and
  so~$c_{s'-1},c_{s'},\tilde{v}$ satisfy the conditions of the lemma, as
  claimed. We set~$v_i=v$ and $v_{i+1}=v'$, and use the diagram shown
  in~\cref{fig:twoleft}~(right). If~$v''$ is eligible in~$G_{i+1}$, that is, the
  triangle~$vv''w'$ is empty of vertices, then we set~$v_{i+2}=v''$ and have a
  diagram for~$G_{i+2}$ for a cost of~$2+\chi\le 3-3\chi$, for~$\chi\le 1/4$.

  Otherwise, by \cref{lem:mainadapt} we inductively obtain a valid diagram~$D$
  for the subgraph of~$G$ induced by taking~$vv''w'$ as an outer triangle
  together with all vertices inside, with~$v''v$ as a starting edge and~$w'$ as
  a last vertex. Then we plug~$D$ into the triangle~$vv''w'$ as shown in
  \cref{fig:twoleft}~(right). All mountains of~$D$ with left endpoint~$v''$
  carry a credit by \ref{i:mountainMoney} for~$D$.  Thus, the resulting diagram
  is extensible. For the costs we have to account for the fact that~$w'$ is
  considered to contribute~$1-\chi$ credits to~$D$, whereas we had already
  accounted for~$w'$ in the diagram for~$G_{i-1}$. On the other hand, the
  edge~$v''w'$ is paid for as a part of~$D$. Thus, the additional costs to
  handle~$v,v',v''$ are~$(1-\chi)+1+\chi=2\le 3-3\chi$, for~$\chi\le 1/3$.
\end{proof}

\begin{figure}[htbp]
  \begin{minipage}[b]{.49\textwidth}
    \centering\includegraphics[page=16]{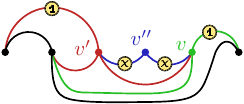}
  \end{minipage}\hfill
  \begin{minipage}[b]{.49\textwidth}
    \centering\includegraphics[page=18]{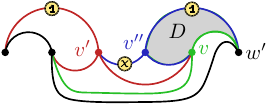}
  \end{minipage}
  \caption{Two vertices~$v,v'$ that have left pivot
    type and~$v''\ne u$ covers the edge~$vv'$.\label{fig:twoleft}}
\end{figure}


To complete the proof of \cref{lem:main,lem:mainadapt} it remains to insert~$u$ along with the set~$\mathcal{V}_u:=V\cap R_{i-1}(u)$ of all vertices inside~$X_1,\ldots,X_{k-1}$, at a cost of~$1-\chi$ credits per vertex. We process these regions from right to left in two phases: In Phase~1, we select a suitable collection~$X_j,\ldots,X_{k-1}$ of regions, for some~$j\in\{1,\ldots,k-1\}$, so that we can insert~$u$ together with all the vertices inside these regions. Then in Phase~2, we process the remaining regions, assuming that~$u$ is already placed on the spine, somewhere to the right. To achieve this we do a case analysis, depending on the four types of regions: left, right, both pivot, or empty. In \cref{app:regions}, we show that in all cases~$u\cup\mathcal{V}_u$ can be inserted as required. 

\section{Triangulations with many degree three vertices}\label{sec:kleetopes}


\begin{restatable}{rtheorem}{thmdegreethree}
    \label{thm:degreethree}
     \renewcommand{\mylink}{\statlink{thm:degreethree}\prooflink{p:degreethree}}
  Let~$G$ be a triangulation with~$n$ vertices, and let~$d$ denote the number of degree three vertices in~$G$. Then~$G$ admits a monotone plane arc diagram with at most~$n-d-4$ biarcs, where every biarc is down-up.
\end{restatable}
\begin{proof}[Proof Sketch]
    Let~$T$ denote the triangulation that results from removing all degree-$3$ vertices from~$G$, i.e., ~$T$ has~$k=n-d$ vertices. We proceed in two steps; see \cref{app:kleetopes} for details. 

\begin{figure}[htbp]
  \begin{minipage}[b]{.49\textwidth}
    \centering\includegraphics[page=44]{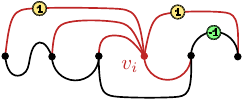}
  \end{minipage}\hfill
  \begin{minipage}[b]{.49\textwidth}
    \centering\includegraphics[page=43]{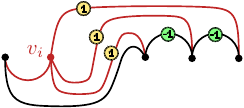}
  \end{minipage}
  \caption{Insert a vertex using at most one credit and make every triangle cross the spine.\label{fig:klee}}
\end{figure}

    \subparagraph*{First step.} We draw $T$ while maintaining Invariants \ref{i:biarcTypes}--\ref{i:biarcMoney} and \ref{i:contour} using the following modifications of our default insertion rules; see \cref{fig:klee}. First, if we insert $v_i$ into a pocket, we always ensure that the leftmost edge incident to $v_i$ is a mountain. Second, if all edges covered by~$v_i$ are mountains, we \emph{push down} the leftmost such mountain $m$, that is, we redraw $m$ and all mountains having the same left endpoint as $m$ into down-up biarcs. Third, instead of assigning credits to covered mountains whose left endpoint remains on the outer face, we immediately transform them into biarcs. Fourth, each vertex aside from $v_1,v_2,v_3,v_n$ contributes $1$ credit to the charging scheme. As a result, the arc diagram of $T$ has at most $n-d-4$ biarcs and all created faces have a non-empty intersection with the spine---{note that the latter property does not follow from from the result by Cardinal et al.~\cite{chktw-adfdht-18}.}

    \subparagraph*{Second step.} We insert each degree-three vertex~$v$ in its containing face~$f$ of~$T$. Using that~$f$ crosses the spine we can place~$v$ there and then realize each edge to a vertex of~$f$ as a proper arcs. Thus, no new biarcs are created in the second step.
\end{proof}

\kleetopes*
\begin{proof}
  Let~$G$ be a Kleetope on~$n$ vertices, and let~$d$ denote the number of degree three vertices in~$G$. By \cref{thm:degreethree} the graph~$G$ admits a monotone plane arc diagram with at most~$n-d-4$ biarcs, where every biarc is down-up. Removing the degree three vertices from~$G$ we obtain a triangulation~$T$ on~$n-d$ vertices, which by Euler's formula has~$2(n-d)-4$ triangular faces. As~$G$ is a Kleetope, it is obtained by inserting a vertex into each of these faces, that is, we have~$n=(n-d)+2(n-d)-4$ and thus~$d=(2n-4)/3$. So there are at most~$n-d-4=(n-8)/3$ biarcs in the diagram.
\end{proof}

\section{Planar 3-Trees}\label{sec:3trees}

For~$3$-trees it is natural to follow their recursive construction sequence and build a corresponding diagram incrementally. A planar~$3$-tree~$G$ is built by starting from a (combinatorial) triangle. At each step we insert a new vertex~$v$ into a (triangular) face~$f$ of the graph built so far, and connect~$v$ to the three vertices of~$f$. Every planar~$3$-tree~$G$ on at least four vertices is~$3$-connected. So its combinatorial embedding is unique, and for each triangle of the abstract graph we know whether it is facial or separating. In the former case, there is exactly one vertex of~$G$ that is adjacent to all vertices of the triangle, in the latter case there are exactly two such vertices. In particular, we can pick any facial triangle to be the starting triangle of our construction sequence 
for~$G$ and become the outer face of our diagram.

Let~$v_1,\ldots,v_n$ be such a construction sequence for~$G$. For~$i\in\{3,\ldots,n\}$, let~$V_i=\{v_1,\ldots,v_i\}$ and~$G_i=G[V_i]$. 
%
%
Each vertex~$v_i$, for~$i\in\{4,\ldots,n\}$, is inserted into a face~$\mathrm{F}(v_i)=uvw$ of~$G_{i-1}$, creating three \emph{child faces} $uvv_i$, $vwv_i$ and $wuv_i$ of~$uvw$ in $G_{i}$. 
We also say that~$v_i$ is the \emph{face vertex} $\mathrm{v}(uvw)$ of face~$uvw$.
We call a face~$f$ of~$G_i$ \emph{active} if it has a face vertex in~$V\setminus V_i$; otherwise, it is \emph{inactive}. The \emph{grand-degree}~$\mathrm{gd}(f)$ is the maximum number of active child faces of~$f$ in all of~$G_3,\ldots, G_n$. Observe that by construction~$\mathrm{gd}(f)\in\{0,\ldots,3\}$ and that $f$ is active for some~$G_i$ if and only if~$\mathrm{gd}(f) > 0$. Similarly, a vertex is a gd-$i$ vertex, for~$i\in\{0,1,2,3\}$, if it is the face vertex of a face~$f$ with~$\mathrm{gd}(f)=i$. For a construction sequence we define its dual \emph{face tree} $\mathcal{T}$ on the faces of all~$G_i$ such that
the root of $\mathcal{T}$ is $v_1v_2v_3$, and
each active face~$uvw$ has three children: 
the faces~$uvz$, $vwz$, and~$wuz$, where~$z=\mathrm{v}(uvw)$. 
%
Note that the leaves of $\mathcal{T}$ are inactive for all~$G_i$. Let us first observe that no biarcs are needed if all faces have small grand-degree. To this end, also recall that $G$ admits a plane proper arc diagram if and only if it is subhamiltonian and planar.

\begin{theorem}\label{thm:gd2}
  Let~$G$ be a planar $3$-tree that has a construction sequence~$v_1,\ldots,v_n$ such that for each face~$f$ in its dual tree $\mathrm{gd}(f)\le 2$. Then~$G$ admits a plane proper arc diagram.
\end{theorem}
\begin{proof}
  We start by drawing the face~$v_1v_2v_3$ as a \emph{drop}, that is, a face where the two short edges are proper arcs on different sides of the spine; see \cref{fig:drops}. Then we iteratively insert the vertices~$v_i$, for~$i=4,\ldots,n$, such that every face that corresponds to an internal vertex of the dual tree~$\mathcal{T}$ is a drop in the diagram~$D_i$ for~$G_i$. This can be achieved because by assumption at least one of the three faces of~$D_i$ created by inserting~$v_i$ is a leaf of~$\mathcal{T}$, which need not be realized as a drop. But we can always realize the two other faces as drops, as shown in \cref{fig:drops}. In this way we obtain a diagram for~$G$ without any biarc. 
\end{proof}

\begin{figure}[htbp]
  \centering\includegraphics[page=1]{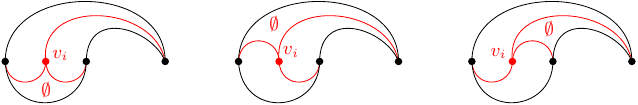}
  \caption{Insert a vertex~$v_i$ into a drop 
  s.t.~any chosen two of the faces created are drops.\label{fig:drops}}
\end{figure}

As~$\mathcal{T}$ is a tree, we can relate the number of internal vertices to the number of leaves.

\begin{lemma}\label{lem:inactivenodes}
  Let $f_d$ denote the number of faces in~$\mathcal{T}$ with grand-degree exactly~$d$, and 
  let $n_{\mathrm{inact}}$ denote the number of face vertices that create inactive faces only. Then~$n_{\mathrm{inact}} \geq 2f_3+f_2$.
\end{lemma}
\begin{proof}
  Consider the rooted tree~$\mathcal{T}'$ obtained by removing all leaves of~$\mathcal{T}$, and observe that the grand-degree in~$\mathcal{T}$ corresponds to the vertex degree in~$\mathcal{T}'$.
\end{proof}

%

We are now ready to describe our drawing algorithm for general planar $3$-trees.

\threetrees*
\begin{proof}
    Our algorithm is iterative and draws $G$ in the sequence prescribed by $\mathcal{T}$. Namely, at each step of our algorithm, we select an arbitrary already drawn face $uvw$ and insert its face vertex $\mathrm{v}(uvw)$, possibly together with the face vertex of a child face. We will consider faces of a particular shape mostly. Consider a face $f=uvw$ such that~$u,v,w$ appear in this order along the spine and~$uw$ forms the upper envelope of~$f$. (There is a symmetric configuration, obtained by a rotation by an angle of~$\pi$ where~$uw$ forms the lower envelope of~$f$.) We say that~$f$ is \emph{ottifant-shaped}\footnote{An \emph{ottifant} is a cartoon abstraction of an elephant designed and popularized by the artist Otto Waalkes. Use of the term \emph{ottifant} with kind permission of Ottifant Productions GmbH.} if it contains a region bounded by a down-up biarc between~$u$ and~$w$, a down-up biarc between~$u$ and~$v$ and a mountain between~$v$ and~$w$; see \cref{fig:ottifant:1}. Note the word ``contains'' in the definition of ottifant-shaped, which allows the actual face to be larger. For instance, the top boundary could be a mountain, but we treat it as if it was a biarc for the purposes of drawing edges; that is, we only connect to~$u$ from below the spine. 
    
    To control the number of biarcs drawn we maintain a charge~$\mathrm{ch}(v)$ for each vertex~$v$. We require additional flexibility from the edge~$vw$ of an ottifant-shaped face~$f=uvw$, which we call the \emph{belly} of~$f$. To this end, we call a mountain~$vw$ \emph{transformable} if it can be redrawn as a down-up biarc for at most~$3/2$ units of charge. (Note that every edge can be drawn as a biarc for only one credit. But in some cases redrawing an edge as a biarc requires another adjacent edge to be redrawn as a biarc as well. Having an extra reserve of half a credit turns out sufficient to cover these additional costs, as shown in the analysis below.)

    %
\ifmhotti
  More specifically, we maintain the following invariants:
    \begin{enumerate}[label=(O\arabic*),left=\labelsep]
        \item\label{inv:otti:1} Each internal active face is ottifant-shaped.
        \item\label{inv:otti:2} If the belly of an active face is a mountain, it is transformable.
        \item\label{inv:otti:3} The sum of the charges of all vertices is at least the number of biarcs drawn.
        \item\label{inv:otti:4} For each vertex $v$ we have $\mathrm{ch}(v)\le\frac34$.
    \end{enumerate}

    \noindent It is easy to see that a drawing $D$ of $G$ has at most $\lfloor \frac34n \rfloor$ biarcs if the invariants hold for $D$. 

\begin{figure}[htbp]
  \centering
  {\captionsetup{singlelinecheck=true}
  \begin{subfigure}[b]{.32\textwidth}
    \centering
    \includegraphics[page=2]{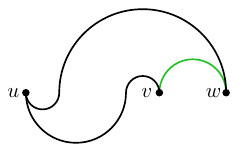}
    \subcaption{}  
    \label{fig:ottifant:1}  
  \end{subfigure}\hfill
  \begin{subfigure}[b]{.32\textwidth}
    \centering
    \includegraphics[page=19]{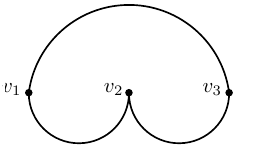}
    \subcaption{}  
    \label{fig:ottifant:2}  
  \end{subfigure}
    \hfill
     \begin{subfigure}[b]{.32\textwidth}
    \centering
    \includegraphics[page=4]{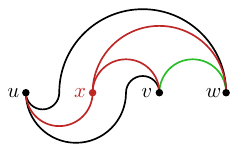}
    \subcaption{}  
    \label{fig:ottifant:3}  
    \end{subfigure}\hfill}
    \caption{(a)~
    An ottifant-shaped face $uvw$, where the long edge is on the top page (green edges are transformable). (b)~Drawing of the initial face $v_1v_2v_3$. (c)~Insertion of a gd-$1$ vertex~$x=\mathrm{v}(uvw)$.}
    \label{fig:ottifant}
\end{figure}

  \subparagraph{Initialization.} We put~$v_1v_2v_3$ on the spine in this order and draw the edges~$v_1v_2$ and~$v_2v_3$ as pockets and~$v_1v_3$ as a mountain; see \cref{fig:ottifant:2}. The invariants \ref{inv:otti:1}--\ref{inv:otti:4} hold.
  

  \subparagraph{Charging rights.} Typically we charge a vertex when it is added to the drawing. But different vertices have different needs. Specifically, we will see that no biarc/charge is used when inserting a gd-$0$ vertex. Therefore, for each gd-$0$ vertex~$v$ we distribute the rights to use the charge of~$v$ among two targets: (1)~the \emph{parent} of~$v$ (i.e., the vertex~$\mathrm{v}(f)$ of the parent~$f$ of~$\mathrm{F}(v)$ in~$\mathcal{T}$)---if it exists---may assign a charge of~$\le 1/4$ to~$v$ and (2)~the so-called \emph{preferred ancestor}~$p(v)$ may assign a charge of~$\le 1/2$ to~$v$.
  %
   Preferred ancestors are determined by selecting an arbitrary surjective map~$p$ from the set of gd-$0$ vertices to the set of gd-$2$ and gd-$3$ vertices. According to \cref{lem:inactivenodes} there exists such a map such that every gd-$2$ is selected at least once and every gd-$3$ vertex is selected at least twice as a preferred ancestor.

  \subparagraph{Iterative step.} We select an arbitrary active face~$f=uvw$, which is  ottifant-shaped by \ref{inv:otti:1}, and insert its face vertex $x:=\mathrm{v}(f)$ into~$f$. 
  Assume w.l.o.g. (up to rotation by an angle of~$\pi$) that~$uw$ forms the top boundary of~$f$. 
  We make a case distinction based on~$\mathrm{gd}(f)$.

  \subparagraph{Case 1: $\mathrm{gd}(f)=0$.} Then all child faces of~$f$ are inactive so that \ref{inv:otti:1} and \ref{inv:otti:2} hold trivially. We insert~$x$ inside~$f$ between~$u$ and~$v$ on the spine, draw the edge~$ux$ as a pocket and~$xv$ and~$xw$ as mountains; see \cref{fig:ottifant:3}. No biarcs are created, so \ref{inv:otti:3}--\ref{inv:otti:4} hold.
  
  \subparagraph{Case 2: $\mathrm{gd}(f)\ge 2$.} We insert~$x$ as in Case $1$, except that~$xv$ is drawn as a biarc rather than as a mountain; see \cref{fig:ottifant:4}. 
  All created child faces are ottifant-shaped \ref{inv:otti:1} and all bellies are transformable \ref{inv:otti:2}. We created one biarc. So to establish \ref{inv:otti:3}--\ref{inv:otti:4} it suffices to set~$\mathrm{ch}(x)=\frac34$ and add a charge of~$\frac14$ to one of the (at least one) gd-$0$ vertices in~$p^{-1}(x)$. 
  
 \begin{figure}[htbp]
  \centering
  {\captionsetup{singlelinecheck=true}
    \centering 
     \begin{subfigure}[b]{.32\textwidth}
    \centering
    \includegraphics[page=3]{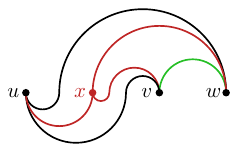}
    \subcaption{}  
    \label{fig:ottifant:4}  
  \end{subfigure}\hfill
  \begin{subfigure}[b]{.32\textwidth}
    \centering
    \includegraphics[page=12]{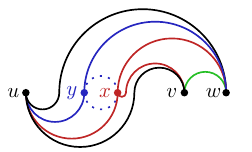}
    \subcaption{}  
    \label{fig:ottifant:9}  
  \end{subfigure}\hfill
  \begin{subfigure}[b]{.32\textwidth}
    \centering
    \includegraphics[page=17]{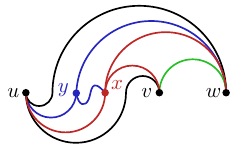}
    \subcaption{}  
    \label{fig:ottifant:10}  
  \end{subfigure}\hfill}
  \caption{Insertion of (a)~a gd-$2$ vertex~$x$; (b)~a gd-$1$ vertex~$y$; (c)~a gd-$2$ vertex~$y$.\label{fig:ottifant-steps:1}}
  \end{figure}

  \subparagraph{Case 3: $\mathrm{gd}(f)=1$.} Then only one of the three child faces of~$f$ is active. If~$uvx$ is the active child face, then we use the same drawing as for a gd-$0$ vertex (see \cref{fig:ottifant:3}) and all invariants hold. However, if one of the other child faces is active, then we cannot use this drawing because~$xw$ is not transformable and~$xvw$ is not ottifant-shaped. 

  So we also consider the face vertex~$y$ of the unique child face~$f'$ of~$x$ and insert both~$x$ and~$y$ into the drawing together. We consider two subcases, according to~$f'$.

  \subparagraph{Case 3A: $f'=uxw$.} If~$\mathrm{gd}(f')=0$, then we can once again use the drawing for a gd-$0$ vertex (see \cref{fig:ottifant:3}) because~$f'$ is ottifant-shaped and none of its child faces are active. 
  
  If~$\mathrm{gd}(f')=1$, then we add first~$x$ as described for a gd-$2$ vertex above (see \cref{fig:ottifant:4}). Then we add~$y$ into~$f'$ and draw all incident edges as proper arcs; the edge~$yx$ can be drawn either as a mountain (if~$uxy$ is the active child face of~$f'$) or as a pocket (otherwise); see \cref{fig:ottifant:9}. In either case, invariants \ref{inv:otti:1}--\ref{inv:otti:2} hold. We added one biarc ($xv$). To establish \ref{inv:otti:3}--\ref{inv:otti:4} we set~$\mathrm{ch}(x)=\mathrm{ch}(y)=\frac12<\frac34$.
  
  Otherwise, we have~$\mathrm{gd}(f')\ge 2$. We first add~$x$ as described above for a gd-$0$ vertex and then~$y$ as a gd-$2$ vertex; see \cref{fig:ottifant:10}. Invariant \ref{inv:otti:1} holds. To establish \ref{inv:otti:2} we have to make the bellies~$xw$ and $uy$ of~$yxw$ and $uyx$, respectively, transformable. To this end, we put~$1/2$ units of charge aside so that both~$xv$ and~$xw$ could be redrawn as biarcs for~$3/2$ units of charge, as required. Moreover, we observe that $uy$ can be transformed into a biarc for $1$ units of charge if necessary as there is no other edge that must be transformed in this scenario. We also added a biarc, namely, $yx$. To establish \ref{inv:otti:3}--\ref{inv:otti:4} we set~$\mathrm{ch}(x)=\mathrm{ch}(y)=\frac34$.

   \subparagraph{Case 3B: $f'=xvw$.} We consider several subcases according to~$\mathrm{gd}(f')$. 
  If~$\mathrm{gd}(f')=0$, we first insert~$x$ as described above for a gd-$2$ vertex and then~$y$ as a gd-$0$ vertex; see \cref{fig:ottifant:5}. Invariants \ref{inv:otti:1}--\ref{inv:otti:2} hold trivially. We used one biarc ($xv$). To establish \ref{inv:otti:3}--\ref{inv:otti:4}, we set~$\mathrm{ch}(x)=\frac34$ and increase~$\mathrm{ch}(y)$ by~$\frac14$. The latter is allowed because~$x$ is the parent of~$y$. 

  \begin{figure}[htbp]
    \centering {\captionsetup{singlelinecheck=true}
     \begin{subfigure}[b]{.32\textwidth}
    \centering
    \includegraphics[page=14]{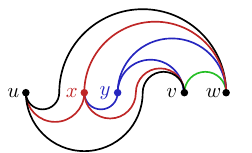}
    \subcaption{}  
    \label{fig:ottifant:5}  
  \end{subfigure}\hfill
  \begin{subfigure}[b]{.32\textwidth}
    \centering
    \includegraphics[page=13]{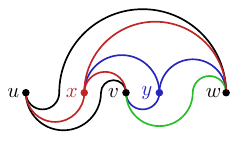}
    \subcaption{}  
    \label{fig:ottifant:6}  
  \end{subfigure}\hfill
  \begin{subfigure}[b]{.32\textwidth}
    \centering
    \includegraphics[page=18]{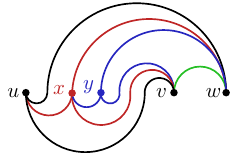}
    \subcaption{}  
    \label{fig:ottifant:7}  
  \end{subfigure}\hfill}
  \caption{Insertion of (a)~a gd-$2$ vertex~$x$; (b)~a gd-$1$ vertex~$y$; (c)~a gd-$2$ vertex~$y$.\label{fig:ottifant-steps:2}}
  \end{figure}

  We use the same drawing if~$\mathrm{gd}(f')=1$ and the (only) active child face of~$f'$ is~$xvy$ or~$xyw$. If~$xvy$ is active, then we set~$\mathrm{ch}(x)=\mathrm{ch}(y)=\frac12<\frac34$ to  establish \ref{inv:otti:3}--\ref{inv:otti:4}. If~$xyw$ is active, then we put~$1/2$ units of charge aside to make~$yw$ transformable and establish \ref{inv:otti:2}. Then we set~$\mathrm{ch}(x)=\mathrm{ch}(y)=\frac34$ to establish \ref{inv:otti:3}--\ref{inv:otti:4}. 

  If~$\mathrm{gd}(f')=1$, then it remains to consider the case that the (only) active child face of~$f'$ is~$yvw$. We transform~$vw$ into a biarc, then  insert~$x$ between~$u$ and~$v$, and finally insert~$y$ between~$v$ and~$w$ on the spine inside~$f$. All edges incident to~$x$ and~$y$ are drawn as proper arcs; see \cref{fig:ottifant:6}. The only active (grand)child face of~$f$ is~$yvw$, and \ref{inv:otti:1}--\ref{inv:otti:2} hold. We have spent~$3/2$ units of charge to transform~$vw$, and we did not create any biarc. Thus, it suffices to set~$\mathrm{ch}(x)=\mathrm{ch}(y)=\frac34$ to establish \ref{inv:otti:3}--\ref{inv:otti:4}. 

  If~$\mathrm{gd}(f')\ge 2$, then we first insert~$x$ between~$u$ and~$v$ and then~$y$ between~$x$ and~$v$ on the spine inside~$f$. Then we draw~$xv$ and~$yv$ as biarcs and the remaining edges as proper arcs such that~$xy$ is a pocket; see \cref{fig:ottifant:7}. Invariants \ref{inv:otti:1}--\ref{inv:otti:2} hold. We created two biarcs ($xv$ and~$yv$). To establish \ref{inv:otti:3}--\ref{inv:otti:4}, we set~$\mathrm{ch}(x)=\mathrm{ch}(y)=\frac34$ and we increase the charge of a vertex in~$p^{-1}(y)$ by~$1/2$.
\fi
%
  It follows that \ref{inv:otti:1}--\ref{inv:otti:4} hold after each step .
\end{proof}

\section{Conclusions}
\label{sec:conclusions}

We proved the first upper bound of the form~$c \cdot n$, with~$c<1$, for the number of monotone biarcs in arc diagrams of planar graphs.
In our analysis, only some cases require $\chi\le 1/5$, indicating a possibility to further refine the analysis to achieve an even better bound.
It remains open whether there exists a ``monotonicity penalty'' in this problem, but we ruled out the probably most prominent class of non-Hamiltonian maximal planar graphs, the Kleetopes, as candidates to exhibit such a phenomenon. It would be very interesting to close the gap between upper and lower bounds, both in the monotone and in the general settings.

\todosc{Side question (for another time): How fast can the optimization problem of biarc minimization be solved for planar graphs of bounded treewidth? E.g., can the approach given here be made optimal?}



\bibliographystyle{plainurl}
\bibliography{abbrv,arcs}



\appendix

\section{Proof of \cref{lem:extend}}

\lemextend*
\begin{proof}\label{Plemextend}  
  The~$\Rightarrow$ direction is a direct consequence of~\ref{co:3}. For the proof of the other implication, let~$v_1,\ldots,v_i$ be a canonical ordering for~$G_i$, let~$\mathcal{C}_i=\{c_e\colon e\in P_\circ(G_i)\}\ne\emptyset$, and let~$v$ be a minimal element of~$\mathcal{C}_i$ (w.r.t.~$\prec$). We claim that~$v_{i+1}:=v$ is eligible. To see this it suffices to show that~$v_1,\ldots,v_{i+1}$ is a canonical ordering for~$G_{i+1}$ with~$V\setminus V_{i+1}\subset F_\circ(G_{i+1})$. Then the claim and the lemma follow by induction on~$n-i$.

  \ref{co:2} trivially holds for all permutations of~$V$ that start with~$v_1,v_2$, where~$v_1v_2$ is an edge of~$C_\circ(G)$. 
  To prove \ref{co:1} and \ref{co:3} we use that~$v$ is a minimal element of~$\mathcal{C}_i$ and our assumption~$v\in V\setminus V_i\subset F_\circ(G_i)$.  Note that~$d_i(v)\ge 2$ (because of the edge~$e\in P_\circ(G_i)$ for which~$v=c_e$). Therefore, the region~$R_i(v)$ is bounded by a cycle of the plane graph~$G$ through~$v$ and~$P_\circ(G_i)$. We claim that~$R_i(v)\cap V=\emptyset$.

  Suppose to the contrary that there exists a vertex~$w\in R_i(v)\cap V$. Then~$w\notin V_i$ because~$G_i$ is biconnected (so~$C_\circ(G_i)$ is a cycle), $w\in F_\circ(G_i)$, and~$G$ is plane. Thus, while~$w\in F_\circ(G_i)$ by the assumption of the implication, $w$~lies in a bounded face~$f$ of~$G_{i+1}$. Then there exists an edge~$xy\in P_\circ(G_i)$ on the boundary~$\partial f$ of~$f$ in~$G_{i+1}$. But~$f$ is not a face of~$G$ because~$w\in f$. So we have~$z=c_{xy}\in V\setminus V_{i+1}$ which, as~$G$ is plane, implies~$R_i(z)\subset R_i(v)$, in contradiction to~$v$ being a minimal element of~$\mathcal{C}_i$. Therefore, there exists no such vertex~$w$ and~$R_i(v)\cap V=\emptyset$, as claimed.

  As~$G$ is plane, $G_{i+1}$ is an induced subgraph, and $R_i(v)\cap V=\emptyset$, all faces of~$G_{i+1}$ in~$R_i(v)$ are also bounded faces of~$G$. Thus, \ref{co:1} holds for~$v_1,\ldots,v_{i+1}$ because~$G$ is internally triangulated. The additional condition~$V\setminus V_{i+1}\subset F_\circ(G_{i+1})$ is implied by~$F_\circ(G_{i+1})=F_\circ(G_i)\setminus\mathrm{cl}(R_i(v_{i+1}))$ and~$R_i(v)\cap V=\emptyset$, where~$\mathrm{cl}(A)$ denotes the closure of~$A$.
\end{proof}

\section{Omitted proofs from \cref{sec:default}}

\lemdefaultApproachone*
\label{PlemdefaultApproach1}
\begin{proof}
  We place~$v_i$ into the rightmost pocket~$p_{\ell}p_r$ it covers and draw all
  edges incident to~$v_i$ as proper arcs. The path $p_{\ell}v_{i}p_r$ is drawn
  as two pockets, all other new edges are drawn as mountains; see
  \cref{fig:naive}. As the pocket~$p_{\ell}p_r$ is not on~$C_\circ(G_i)$, we can
  take and spend the~$\chi$ credits on it. 
  If~$d_i=2$, then we place~$\chi$
  credits on each of the two pockets incident to~$v_i$ so as to
  establish~\ref{i:pocketMoney}, for a cost of~$\chi\le 1-\chi$. It is easily
  checked that the invariants are maintained, which completes the proof in this
  case.

  \begin{figure}[htbp]
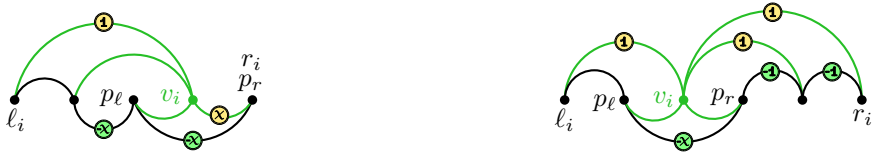

    \centering%
    \begin{minipage}[b]{.45\textwidth}
      \centering
      \includegraphics[page=3]{arcDiagramsFigures}
    \end{minipage}\hfill
    \begin{minipage}[b]{.45\textwidth}
      \centering
      \includegraphics[page=1]{arcDiagramsFigures}
    \end{minipage}
    \caption{Inserting a vertex $v_i$ into a pocket, using $1-\chi$ credits
      (\cref{lem:defaultApproach1}).\label{fig:naive}}
  \end{figure}

  It remains to consider the case $d_i\ge 3$. Here, we describe how to assign credits to the edges incident to $v_i$. First, consider edges $v_iu$ with $u \notin \{p_\ell,p_r\}$. Note that the edge~$v_iu$ is drawn as a mountain. First, assume that  $u$ lies
  to the right of~$p_r$ on~$P_\circ(G_{i-1})$, then~$v_i$ covers the edge~$e_u$
  of~$P_\circ(G_i)$ whose right endpoint is~$u$. By the choice of~$p_{\ell}p_r$
  (as the rightmost pocket covered by~$v_i$), the edge~$e_u$ is a mountain,
  which by \ref{i:mountainMoney} carries one credit. As~$e_u$ is not
  on~$C_\circ(G_i)$, we can transfer this credit to the edge~$v_iu$, so as to
  satisfy~\ref{i:mountainMoney} for~$v_iu$ since for two such edges $uv_i$ and $u'v_i$ with $u$ and $u'$ to the right of $p_r$ we have that $e_u \neq e_{u'}$. Second,  consider case where the vertex~$u$ lies to
  the left of~$p_\ell$ on~$P_\circ(G_{i-1})$. Note that the left endpoint
  of~$uv_i$ is not on~$C_\circ(G_i)$, unless~$u=\ell_i$. Therefore, it suffices
  to pay one credit in total and place it on~$\ell_iv_i$ to
  establish~\ref{i:mountainMoney} for the resulting diagram.

  As no biarc is created by the insertion of~$v_i$, the only remaining possible
  sources of costs are pockets of~$P_\circ(G_i)$ incident
  to~$v_i$. If~$\ell_i=p_\ell$, then there is such a pocket to the left
  of~$v_i$, and if~$r_i=p_r$, then there is such a pocket to the right of~$v_i$
  on~$P_\circ(G_i)$. As~$d_i\ge 3$, we face at most one of these pockets. To pay
  the~$\chi$ credits for this pocket (if it exists) to
  establish~\ref{i:pocketMoney} we can use the~$\chi$ credits from the
  pocket~$p_{\ell}p_r$, which is covered by~$v_i$.

  Overall, we pay at most one credit to insert~$v_i$, which proves the first statement
  of the lemma. To prove the second statement, we need to argue how to
  save~$\chi$ credits if~
  $\mathrm{pr}(v_i)\ne\;\frown\smile$. If there is no pocket incident to~$v_i$ in~$P_\circ(G_i)$, then we save
  the~$\chi$ credits that we accounted for such a pocket, which completes the
  proof in this case.  Thus, it remains to consider the two
  cases~$\ell_i=p_\ell$ and~$r_i=p_r$ only.

  If~$\ell_i=p_\ell$, then we save the one credit that we accounted for the
  mountain~$\ell_iv_i$ in the previous analysis as $\ell_iv_i$ is actually a pocket here. So, the overall costs are
  zero in this case. Otherwise, we have~$r_i=p_r$. If~$d_i=3$, then we
  have~$\mathrm{pr}(v_i)=\;\smile\smile$ as we explicitly exclude the profile $\frown\smile$. So~$v_i$ covers two pockets and we can
  take the~$\chi$ credits from both, whereas we spend only~$\chi$ credits on the
  pocket~$v_ip_r$. Thus, the overall costs are~$1-\chi$, which completes the
  proof in this case. The situation is similar in the remaining case~$d_i\ge 4$
  because~$v_i$ covers at least two edges of~$P_\circ(G_i)$ to the left
  of~$p_\ell$. In particular, at least one edge~$e$ to the left of~$p_\ell$ is
  covered by~$v_i$ such that the left endpoint of~$e$ is not
  on~$C_\circ(G_i)$. Either~$e$ is a mountain, in which case we can take the one
  credit it carries, or it is a pocket, and we can take the~$\chi$ credits it
  carries. Either way, we gain at least~$\chi$ credits, for overall costs of at
  most~$1-\chi$.
\end{proof}

\lemdefaultApproachtwo*\label{PlemdefaultApproach2}
\begin{proof}
  We push down the rightmost mountain~$mr_i$ in~$\mathrm{pr}(v_i)$ and
  place~$v_i$ above it. By \emph{push down} we mean that each mountain
  in~$G_{i-1}$ with left endpoint~$m$ (there is only one such mountain
  on~$P_\circ(G_{i-1})$, but there may be many more underneath) is transformed
  into a down-up biarc; see \cref{fig:mountains}. The costs for these biarcs, so
  as to maintain~\ref{i:biarcMoney}, are covered by the credits that each
  mountain whose left endpoint is on~$C_\circ(G_i)$ carries according to
  \ref{i:mountainMoney}.

  The insertion of~$v_i$ creates a new pocket and~$d_i-2$ new mountains. Out of
  these new edges only two mountains, namely~$\ell_iv_i$ and~$v_ir_i$, appear
  on~$P_\circ(G_i)$. Therefore, two credits suffice to establish the
  invariants. As~$v_i$ covers~$d_i-1$ mountains, it covers~$d_i-2$ left
  endpoints of mountains from~$P_\circ(G_{i-1})$. One of these mountains is
  pushed down, consuming the credit it carries. But the at least~$d_i-3$ credits
  on the remaining~$d_i-3$ mountains are now free to be used. Thus, the overall
  costs of inserting~$v_i$ as described are at most~$2-(d_i-3)=5-d_i$.
\end{proof}
 
\begin{figure}[htbp]
  \centering%
  \includegraphics[page=7]{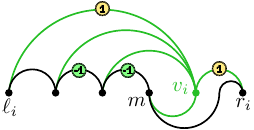}
  \caption{Inserting a vertex~$v_i$ into mountains, using $5-d_i$ credits
    (\cref{lem:defaultApproach2}).\label{fig:mountains}}
\end{figure}

\section{\cref{lem:main,lem:mainadapt} hold for~$i=n$}\label{sec:lemiisn}

A special case arises if~$v=v_n$ is the last vertex of the canonical
ordering. Then~$i=n$, $d_i\ge 3$, and~$v_n$ is the only vertex
in~$\mathcal{E}_{i-1}$. To complete the proof of \cref{lem:main} in this case,
we insert~$v$ as shown in \cref{fig:lastvertex} and observe that the insertion
costs are at most~$1+\chi$ in all cases. The extra costs of at most~$2\chi$
compared to the regular costs of~$1-\chi$ per vertex are taken care of by the~ 
$+\xi$ term in the costs bound of \cref{lem:main}. These are the only cases where we need~$\xi>0$; that is, we actually prove the following, stronger version of \cref{lem:main}.

\begin{lemma}\label{lem:mainstrong}
  \cref{lem:main} holds with~$\xi=0$ or~$D$ can be extended to an extensible arc
  diagram~$D'$ for~$G\setminus\{v_n\}$
  with~$\mathrm{cost}(D')\le\mathrm{cost}(D)+(n-i-1)(1-\chi)$ such that~$v_n$ is
  problematic for~$D'$.
\end{lemma}

\begin{figure}[htbp]
  \centering
  {\captionsetup{singlelinecheck=true}
  \begin{minipage}[b]{.29\textwidth}
    \centering
    \includegraphics[page=32]{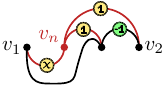}
    \subcaption{$\mathcal{T}(3,\frown^2)$}  
  \end{minipage}\hfill
  \begin{minipage}[b]{.37\textwidth}
    \centering
    \renewcommand{\thelinenumber}{}%
    \includegraphics[page=33]{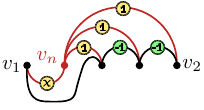}
    \subcaption{$\mathcal{T}(4,\frown^3)$}  
  \end{minipage}\hfill
  \begin{minipage}[b]{.29\textwidth}
    \centering
    \renewcommand{\thelinenumber}{}%
    \includegraphics[page=34]{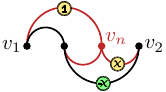}
    \subcaption{$\mathcal{T}(3,\frown\smile)$}  
  \end{minipage}\hfill}
  \caption{Inserting a final problematic vertex~$v_n$ for a cost of~$\le 1+\chi$.\label{fig:lastvertex}}
\end{figure}

To complete the proof of \cref{lem:mainadapt} for~$i=n$, we insert~$v$ as shown
in \cref{fig:lastvertexadapt}.

\begin{figure}[htbp]
  \centering
  {\captionsetup{singlelinecheck=true}
  \begin{minipage}[b]{.29\textwidth}
    \centering
    \includegraphics[page=36]{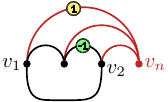}
    \subcaption{$\mathcal{T}(3,\frown^2)$}  
  \end{minipage}\hfill
  \begin{minipage}[b]{.37\textwidth}
    \centering
    \renewcommand{\thelinenumber}{}%
    \includegraphics[page=37]{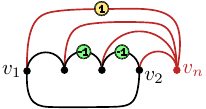}
    \subcaption{$\mathcal{T}(4,\frown^3)$}  
  \end{minipage}\hfill
  \begin{minipage}[b]{.29\textwidth}
    \centering
    \renewcommand{\thelinenumber}{}%
    \includegraphics[page=38]{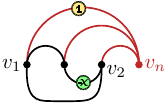}
    \subcaption{$\mathcal{T}(3,\frown\smile)$}  
  \end{minipage}\hfill}
  \caption{Inserting a final problematic vertex~$v_n$ for a cost of~$\le 1-\chi$.\label{fig:lastvertexadapt}}
\end{figure}

\section{Omitted Proofs from \texorpdfstring{\cref{sec:nondefault}}{Section 4}}

\lemu*
\begin{proof}\label{Plemu}
  We know that~$\mathcal{E}_{i-1}\ne\emptyset$ and that all vertices
  in~$\mathcal{E}_{i-1}$ are problematic. Assume for the sake of a contradiction
  that for every~$v\in\mathcal{E}_{i-1}$ we
  have~$\mathrm{pc}(v)\in\mathcal{E}_{i-1}$. Then there exists a cyclic
  sequence~$u_0,\ldots,u_k$ of eligible vertices, for~$k\ge 1$, such
  that~$\mathrm{pc}(u_j)=u_{(j+1)\,\mathrm{mod}\,k}$, for all~$0\le j\le k$. We may assume that all vertices~$u_0,\ldots,u_k$ have left pivot type as otherwise we can apply
  \cref{lem:samepivot}. 

  Every edge of~$P_\circ(G_{i-1})$ is covered by at most one vertex
  from~$\mathcal{E}_{i-1}$, and conversely every vertex in~$\mathcal{E}_{i-1}$
  covers some subpath of at least one consecutive edge(s)
  of~$P_\circ(G_{i-1})$. Thus, we can order the vertices in~$\mathcal{E}_{i-1}$
  from left to right according to the part of~$P_\circ(G_{i-1})$ they
  cover. Without loss of generality let~$u_0$ be the leftmost vertex
  among~$u_0,\ldots,u_k$, and let~$\ell$ be the leftmost neighbor of~$u_0$
  on~$P_\circ(G_{i-1})$. Then by the left-to-right order the edges
  of~$P_\circ(G_{i-1})$ covered by~$u_1$ are to the right of the edges
  of~$P_\circ(G_{i-1})$ covered by~$u_0$. At the same time~$u_1$ is adjacent
  to~$\ell$ because~$u_1=\mathrm{pc}(u_0)$. It follows
  that~$R_{i-1}(u_1)\supset R_{i-1}(u_0)$, which by \cref{lem:eligible} is in
  contradiction to~$u_1\in\mathcal{E}_{i-1}$.
\end{proof}

\lemonepivot*
\begin{proof}\label{Plemonepivot}
  Let~$v\in X_j\cap\mathcal{E}_{i-1}$ with~$\mathrm{pc}(v)=u$. Then~$u$ is
  adjacent to~$\mathrm{p}(v)$ in~$G$. As~$u$ has only two neighbors
  on~$P_\circ(G_{i-1})\cap\partial X_j$, 
  we have~$\mathrm{p}(v)\in\{w_j,w_{j+1}\}$. So, if~$v$ has left pivot type,
  then~$\mathrm{p}(v)=w_j$ and~$v$ is the unique vertex that covers the edge
  of~$P_\circ(G_{i-1})$ whose left endpoint is~$w_j$. Else~$v$ has right pivot
  type, $\mathrm{p}(v)=w_{j+1}$, and $v$ is the unique vertex that covers the
  edge of~$P_\circ(G_{i-1})$ whose right endpoint is~$w_{j+1}$.
\end{proof}

\lemrightpivot*
\begin{proof}\label{Plemrightpivot}
  For every~$v\in X_j$, we have~$\mathrm{pc}(v)\in X_j\cup\{u\}$ by
  planarity. Therefore, by the choice of~$u$ as a minimal element
  of~$\mathcal{U}$, we
  have~$\mathrm{pc}(v)\in\mathcal{E}_{i-1}\cup\{u\}$. If~$v$ has right pivot
  type, then by \cref{lem:samepivot} we
  have~$\mathrm{pc}(v)\notin\mathcal{E}_{i-1}$ and, therefore,
  $\mathrm{pc}(v)=u$. Now the statement follows from \cref{lem:onepivot}.
\end{proof}

\lemleftpivot*
\begin{proof}\label{Plemleftpivot}
  For every~$x\in Q$, we have~$\mathrm{pc}(x)\in X_j\cup\{u\}$ by
  planarity. Thus, by the choice of~$u$ (as a minimal element of~$\mathcal{U}$)
  either~$\mathrm{pc}(x)=u$ or~$x'=\mathrm{pc}(x)\in\mathcal{E}_{i-1}$. By
  \cref{lem:onepivot} the former case applies to at most one vertex of~$Q$. In
  the latter case we may assume that~$x'\in
  Q$ as otherwise we can apply \cref{lem:samepivot}. Each~$y\in\mathcal{E}_{i-1}$ covers a subpath~$\sigma(y)$
  of~$P_\circ(G_{i-1})$ and has no other neighbors
  on~$P_\circ(G_{i-1})$. As~$x'=\mathrm{pc}(x)$, we know that~$x'$ is adjacent
  to~$\mathrm{p}(x)$, which is the left endpoint of~$\sigma(x)$;
  thus~$\mathrm{p}(x)$ is also the right endpoint of~$\sigma(x')$. Therefore, we
  can order the vertices in~$Q$ from left to right, according to the order of
  the corresponding paths~$\sigma(\cdot)$ on~$P_\circ(G_{i-1})$. For the
  leftmost vertex~$x_1$ in this order, we must have~$\mathrm{pc}(x_1)=u$.
\end{proof}

\lemregions*
\begin{proof}\label{Plemregions}
  Assume for a contradiction that~$c_e\ne u$ and~$c_e\notin\mathcal{E}_{i-1}$.
  As~$c_e\ne u$, by planarity~$c_e\in X_j$ and, theferore,
  $R_{i-1}(c_e)\subsetneq R_{i-1}(u)$ and~$c_e\prec
  u$. As~$c_e\notin\mathcal{E}_{i-1}$, by \cref{lem:eligible} there exists a
  vertex~$v\in R_{i-1}(c_e)\cap\mathcal{E}_{i-1}$. By
  planarity~$v'=\mathrm{pc}(v)\in R_{i-1}(c_e)\cup\{c_e\}\subsetneq R_{i-1}(u)$,
  and by the choice of~$u$ (as a minimal element of~$\mathcal{U}$) we
  have~$v'\in\mathcal{E}_{i-1}$. In particular, as~$c_e\notin\mathcal{E}_{i-1}$,
  we have~$v'\ne c_e$. By \cref{lem:rightpivot} both~$v$ and~$v'$ have left
  pivot type. Thus, by \cref{lem:leftpivot} there is a sequence~$x_1,\ldots,x_q$
  of eligible vertices, with~$x_{q-1}=v'$ and~$x_q=v$, such
  that~$x_h=\mathrm{pc}(x_{h+1})$, for all~$1\le h\le q-1$,
  and~$\mathrm{pc}(x_1)=u$. In particular, we have~$x_1\notin R_{i-1}(c_e)$
  because~$u\notin R_{i-1}(c_e)\cup\{c_e\}$. Let~$h\ge 1$ be maximal such
  that~$x_h\notin R_{i-1}(c_e)$, and note that~$1\le h\le q-2$.
  Then~$x_{h+1}\in R_{i-1}(c_e)$ and, therefore,
  $x_h=\mathrm{pc}(x_{h+1})\in R_{i-1}(c_e)\cup\{c_e\}$. It follows
  that~$x_h=c_e$, which, in particular, implies that~$c_e\in\mathcal{E}_{i-1}$,
  a contradiction.
\end{proof}

\section{Processing regions}
\label{app:regions}

Using the insights on the type and structure of eligible vertices within the
regions covered by our selected ``minimally noneligible''
vertex~$u\in\mathcal{U}$ that we have developed in \cref{sec:nondefault} 
we can
now describe how to handle the generic case. It consists of processing~$u$ along
with all regions~$X_1,\ldots,X_{k-1}$ covered by~$u$, thereby
adding~$\nu:=|R_{i-1}(u)\cap V|+1$ vertices to the diagram. So our main goal is
to extend the given extensible arc diagram for~$G_{i-1}$ to an extensible (for
\cref{lem:main}) or at least valid (for \cref{lem:mainadapt}) arc diagram
for~$G_{i-1+\nu}$.

We can classify the regions covered by~$u$ into four different types. Each
region is either \emph{empty}, \emph{left pivot}, \emph{right pivot}, or
\emph{both pivot}---depending on whether it contains no vertices of~$G$, or at
least one eligible vertex of left, right, or both pivot types, respectively. An
empty region has a unique edge of~$P_\circ(G_{i-1})$ on its boundary; depending
on whether this edge is a pocket or mountain we call the corresponding region an
\emph{empty pocket} or an \emph{empty mountain}, respectively.

We proceed in several steps. As a general rule, we process~$X_{k-1},\ldots,X_1$
in this order from right to left. When processing~$X_j$ we assume that~$X_j$ is
not empty and that~$u$ and all edges and vertices inside or on the boundary
of~$X_h$, for all~$h>j$, are placed already; specifically, the edge~$uw_{j+1}$,
which is shared between~$X_{j+1}$ and~$X_j$, is drawn already, and it is already
paid for. 

\begin{enumerate}[resume*=inv]\setlength{\itemindent}{\labelsep}
\item\label{i:uw} The region~$X_j$ is not empty. If~$uw_{j+1}$ is a mountain,
  then it carries~$1-\chi$ credits, and if~$uw_{j+1}$ is a pocket, then it
  carries~$2\chi$ credits.
\end{enumerate}

As an initialization we process some regions~$X_{j+1},\ldots,X_{k-1}$ so as to
establish \ref{i:uw} for~$X_j$. Note that there exists a region~$X_j$,
with~$1\le j<k$, that is nonempty because~$u\notin\mathcal{E}_{i-1}$. Moreover, in the following procedure, if all regions $X_{h'}$ with $h' \in \{j-1,\ldots,j'+1\}$ are empty for some $j'$, we process $X_j,X_{j-1},\ldots,X_{j'+1}$ together so that the next region $X_{j'}$ to be processed is non-empty again.

\subsection{Initialization: Placing~u and
  selecting~X${}_\mathbf{j}$}\label{sec:placement}

\subparagraph{Special case in the proof of \cref{lem:mainadapt}: $u=v_n$.} In this case,  the placement
of~$u$ is determined, as~$u$ must be the rightmost vertex on the
spine. As~$w_k=v_2$, we have to ensure that the edge~$uw_k$ is not drawn as a pocket. 
Let~$X_j$ be the rightmost region that is not empty. We
place~$u$ as the rightmost vertex on the spine and draw all
edges~$uw_k,\ldots,uw_{j+1}$ as mountains and put~$1-\chi$ credits, paid by the
new vertex~$u$, on~$uw_{j+1}$ so as to establish \ref{i:uw} for~$X_j$.


\subparagraph{General cases.} In all other cases, we have to place~$u$ somewhere between~$v_1$ and~$v_2$ on
the spine. To this end, we will select a region~$X_j$, for some~$1\le j<k$, and
place~$u$ as a part of processing~$X_j$.
%
%
We start with the rightmost region~$X_{k-1}$ and work our way from there to the
left. We may suppose without loss of generality that~$X_{k-1}$ is not an empty
mountain. To see this, suppose that~$X_{k-1}$ is an empty mountain. Then we
continue as if~$X_{k-2}$ was the rightmost region. Once all regions are
processed, we add the edge~$uw_k$ to the diagram as a mountain. The costs can be
paid for by a mountain in~$X_{k-1}$ whose left endpoint is covered by~$u$.


If~$X_{k-1}$ is an empty pocket, then we place~$u$ into this pocket. Let~$X_j$
be the rightmost region that is not empty. We pay~$\chi$ credits for the
pocket~$uw_k$, which can be paid for using the~$\chi$ credits on the pocket
of~$X_{k-1}$. Then to establish \ref{i:uw} we have to pay~$1-\chi$ credits
for~$uw_{j+1}$, which is exactly what the new vertex~$u$ provides. So it remains
to consider the case~$X_{k-1}\cap V\ne\emptyset$ only.  We distinguish three
subcases according to the type of~$X_{k-1}$. In all of them, we consider the
plane graph~$G'=G[V_{i-1}\cup X_{k-1}\cup\{u\}]$.

\subparagraph{X$_{\mathbf{k-1}}$ is a left pivot region.}
Then, by \cref{lem:manyleft} we may assume that at most one face of~$G'$ may contain other vertices of~$G$, namely the
triangle~$\Delta=uc_sw_k$ (shaded in figures).

If~$s=1$, then we insert~$c_1$ by pushing down the leftmost mountain it covers
and place~$u$ into the pocket to the left of~$c_1$, see
\cref{fig:allleft}~(left). The figure shows the case that~$c_1$
is~$\mathcal{T}(2,\frown)$; if~$c_1$ covers more mountains to the right, then
the additional mountain(s) at~$c_1$ can be paid for using the credit(s) on the
mountains that are covered, see \cref{fig:allleft}~(middle).

\begin{figure}[htbp]
  \begin{minipage}[b]{.3\textwidth}
    \centering\includegraphics[page=23]{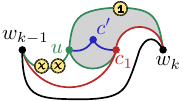}
  \end{minipage}\hfill
   \begin{minipage}[b]{.31\textwidth}
    \centering\includegraphics[page=20]{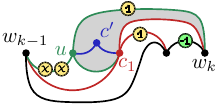}
  \end{minipage}\hfill
   \begin{minipage}[b]{.38\textwidth}
    \centering\includegraphics[page=17]{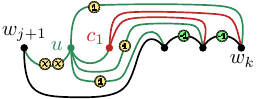}
  \end{minipage}
  \caption{Only one vertex in~$X_{k-1}$ and it has left pivot
    type.\label{fig:allleft}}
\end{figure}

If~$\Delta\cap V=\emptyset$, then the costs are~$1+2\chi\le 2(1-\chi)$,
for~$\chi\le 1/4$, which the two new vertices~$u$ and~$c_1$ can pay. This
suffices to establish \ref{i:uw} for~$X_{k-2}$ if~$X_{k-2}$ is
nonempty. If~$X_{k-2}$ is empty and~$\Delta\cap V=\emptyset$, then we use the
following diagram instead. Let~$j$ be minimal such that all
of~$X_{j+1},\ldots,X_{k-2}$ are empty. We place~$u$ into the unique edge of~$X_{j+1}$
on~$P_\circ(G_{i-1})$, pushing it down if it is a mountain, and add all
edges~$uw_{j+1},\ldots,uw_{k-1}$ as proper arcs. Then we push down all mountains
with left endpoint~$u$ and place~$c_1$ to the right of~$u$. Finally, add~$uw_k$
as a mountain; see \cref{fig:allleft}~(right). Since by assumption, $c_1$ has left pivot type, $u$ covers only mountains. Thus,  the costs for the biarcs at~$u$
can be paid for by the mountains covered by~$u$, if~$j+1<k-2$, plus by the at
least one mountain covered by~$c_1$. So to establish \ref{i:uw} for~$X_j$ we
only have to pay~$2\chi$ credits for the pocket to the left of~$u$ and one
credit for~$uw_k$. This amounts to~$1+2\chi\le 2(1-\chi)$, for~$\chi\le 1/4$,
which the two new vertices~$u$ and~$c_1$ can pay. This approach also works in
case~$j=0$, we even paid~$\chi$ credits too much for~$uw_1$.

Else we have~$\Delta\cap V\ne\emptyset$, and using \cref{lem:mainadapt} we
inductively obtain a valid diagram~$D$ for the subgraph of~$G$ induced by
taking~$\Delta$ as an outer triangle together with all vertices inside,
with~$uc_1$ as a starting edge and~$w_k$ as a last vertex. Then we plug~$D$
into~$\Delta$. If the edge~$c_1w_k$ is drawn as a biarc in~$D$, then we push
down all mountains with left endpoint~$c_1$ (if any exist) to make room. This is
where we need the credits on these mountains if~$d_{i-1}(c_1)>2$. All mountains
of~$D$ with left endpoint~$u$ carry a credit by \ref{i:mountainMoney}
for~$D$. As for the costs, let~$c'$ be the vertex that covers~$uc_1$ in~$D$. We
pay~$2\chi$ credits to initialize the pockets incident to~$c'$ in~$D$. We also
have to account for the fact that~$w_k$ is considered to contribute~$1-\chi$
credits to~$D$, whereas we had already accounted for~$w_k$ in~$G_{i-1}$. In
return the edge~$uw_k$ is paid for as a part of~$D$. Finally, we have to
place~$2\chi$ credits on~$w_{k-1}u$ to establish \ref{i:uw} for~$X_{k-2}$
if~$X_{k-2}$ is nonempty. Otherwise, let~$j$ be minimal such that all
of~$X_{j+1},\ldots,X_{k-2}$ are empty and add all
edges~$uw_{j+1},\ldots,uw_{k-1}$ as mountains. This costs one credit, for the
leftmost mountain~$uw_{j+1}$. So in any case the costs are at
most~$2\chi+(1-\chi)+1=2+\chi\le 3(1-\chi)$, for~$\chi\le 1/4$, which the three
new vertices~$u,c_1,c'$ can pay. Either we have established \ref{i:uw} for
some~$X_j$ or, if~$j=0$, that is, if~$X_{k-1}$ is the only nonempty region, then
this step is complete.

It remains to consider the case~$s>1$. If~$s\ge 3$, then for each vertex~$c_h$,
with~$h\ne s-1$, we push down the leftmost mountain it covers. Then we place
first~$c_{s-1}$ and then~$u$ into the pocket to the left of~$c_s$; see
\cref{fig:allleftmore}~(left).  If~$s=2$, then we push down the leftmost
mountain covered by~$c_1$ and place first~$u$ and then~$c_2$ into the pocket to
the left of~$c_1$; see \cref{fig:allleftmore}~(right). By \ref{i:mountainMoney}
the costs for the biarcs created can be paid for by using the credits on the
mountains that are pushed down. If there are any mountains with left
endpoint~$c_s$ other than~$c_sw_k$, they can be paid for using the credits on
the mountains covered by~$c_s$ to the right. We pay at most one credit for the
edge~$uw_{k-1}$.

\begin{figure}[htbp]
  \begin{minipage}[b]{.56\textwidth}
    \centering\includegraphics[page=21]{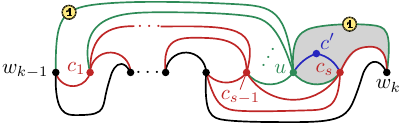}
  \end{minipage}\hfill
   \begin{minipage}[b]{.43\textwidth}
    \centering\includegraphics[page=22]{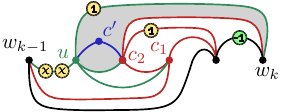}
  \end{minipage}
  \caption{All $s\ge 2$ vertices have left pivot type
    in~$X_{k-1}$.\label{fig:allleftmore}}
\end{figure}

If~$\Delta\cap V=\emptyset$, then we also pay one credit for~$uw_k$, for overall
costs of at most~$2\le 3(1-\chi)\le(s+1)(1-\chi)$, for~$\chi\le 1/3$. Else we
have~$\Delta\cap V\ne\emptyset$, and using \cref{lem:mainadapt} we inductively
obtain a valid diagram~$D$ for the subgraph of~$G$ induced by taking~$\Delta$ as
an outer triangle together with all vertices inside, with~$uc_s$ as a starting
edge and~$w_k$ as a last vertex. Then we plug~$D$ into~$\Delta$. Regarding the
costs we argue as above in the case~$s=1$ to bound them
by~$2\chi+(1-\chi)+1=2+\chi\le 4(1-\chi)\le(s+2)(1-\chi)$, for~$\chi\le 2/5$.

Note that in all cases above we accounted for a cost of one credit for the
edge~$uw_{k-1}$ (even though for~$s=2$ we would have to pay~$2\chi$ credits
only). Therefore, for any sequence~$X_{k-2},\ldots,X_{j+1}$ of empty regions, we
can afford to add the edges~$uw_{k-2},\ldots,uw_{j+2}$ and put one credit
on~$uw_{j+2}$ so as to establish \ref{i:uw} for~$X_j$ if it is nonempty, or even
complete this step in case~$j=0$ (that is, if~$X_{k-1}$ is the only nonempty
region).

\subparagraph{X$_{\mathbf{k-1}}$ is a right pivot region.}
If~$k=2$ or if~$X_{k-2}$ is nonempty, then we push down the mountain covered by~$c_1$ and then place~$u$ into the pocket to the left of~$c_1$, see \cref{fig:onlyright}~(left).  If~$\Delta\cap V=\emptyset$, then the costs to
either finish this step (if~$k=2$) or establish \ref{i:uw} for~$X_{k-2}$ are~$1+\chi\le 2(1-\chi)$, for~$\chi\le 1/3$,
which the two new vertices~$u$ and~$c_1$ can pay. Otherwise, using
\cref{lem:mainstrong} we inductively obtain an extensible diagram~$D$ for the
subgraph~$G_\Delta$ of~$G$ induced by taking~$\Delta$ as an outer triangle
together with all vertices inside, with~$w_{k-1}c_1$ as a starting edge and~$u$
as a last vertex. Then we plug~$D$ into~$\Delta$ and add the
edge~$uw_k$. Let~$c'$ be the vertex that covers~$w_{k-1}c_1$ in~$D$. The costs
are~$2\chi$ credits to initialize the two pockets incident to~$c'$ in~$D$, one
credit for the mountain~$uw_k$, and possibly an additional~$\chi$ credits if the
edge~$w_{k-1}u$ is a pocket in~$D$. To compensate we may take the~$\chi$ credits
on the pocket covered by~$c_1$. Thus, by \cref{lem:mainstrong} the costs to
add~$c_1$ and~$c'$ are at most~$1+2\chi\le 2(1-\chi)$, for~$\chi\le 1/4$, or
there exists an appropriate extensible diagram~$D'$ for~$G_\Delta\setminus\{u\}$
for which~$u$ is problematic. In the latter case, we just plug~$D'$
into~$\Delta$. Then, we distinguish two cases.

If~$u$ has left pivot type in~$G_\Delta$, we push down the leftmost arc covered
by~$u$ in~$D'$ to place~$u$ there. In this case we pay~$2\chi$ credits to initialize the pockets incident to~$c'$. Then we put~$2\chi$ credits on the pocket~$w_{k-1}u$ and two credits on mountains with left endpoint~$u$. There could be three or four mountains with left endpoint~$u$, but any but the first and the last can be paid using the credit on a corresponding mountain covered by the insertion of~$u$. Finally, we can take the~$\chi$ credits on the pocket covered by~$c_1$. So the costs to add~$c_1,c',u$ and establish \ref{i:uw}
for~$X_{k-2}$ are at most~$4\chi+2-\chi=2+3\chi$. This is too much by~$\chi$ because in order to be upper bounded by~$3(1-\chi)$ we would need~$\chi\le 1/6$. However, recall that either~$k=2$ or~$X_{k-2}$ is nonempty by assumption. If~$k=2$, then there is no need to place~$2\chi$ credits on~$w_{k-1}u$ and we can take the missing~$\chi$ credits from there. Otherwise, we undo the insertion of~$u$ but keep the drawing~$G_\Delta\setminus\{u\}$. Next, we pretend that~$X_{k-2}$ is the rightmost region and process it accordingly, as described in this section. Doing so also places~$u$, somewhere to the left of~$w_{k-1}$. Finally, in order to incorporate~$X_{k-1}$ we add the missing edges to~$u$ as mountains and put one credit on each of them. The credits for those mountains that cover~$G_\Delta\setminus\{u\}$ can be taken from the mountains of~$G_\Delta\setminus\{u\}$ that are covered by~$u$. We need to pay one credit for the mountain~$uw_k$ only. In addition, to insert~$c_1$ and~$c'$ we pay~$2\chi$ credits to initialize the pockets incident to~$c'$, but we can take the~$\chi$ credits from the pocket covered by~$c_1$. Thus, these costs are~$1+2\chi-\chi=1+\chi$, and we can afford to pay another~$\chi$ credits, to cover the missing~$\chi$ credits in case that we end up in this very same case when processing~$X_{k-2}$. So in total we account for~$1+2\chi\le 2(1-\chi)$ credits to insert~$c_1$ and~$c'$, for~$\chi\le 1/4$.

Otherwise, the vertex~$u$ has right pivot type in~$G_\Delta$ and we just place it into the pocket of~$D'$ it covers. We pay~$2\chi$ credits to initialize the pockets incident to~$c'$, which can be paid by the~$\chi$ credits each from the pockets covered by~$c_1$ and~$u$. Then we put~$1-\chi$ credits on the mountain~$w_{k-1}u$ and one credit on~$uw_k$. So the costs to add~$c_1,c',u$ and establish \ref{i:uw}
for~$X_{k-2}$ are at most~$2-\chi\le 3(1-\chi)$, for~$\chi\le 1/4$.

\begin{figure}[htbp]
  \hfill\begin{minipage}[b]{.49\textwidth}
    \centering\includegraphics[page=28]{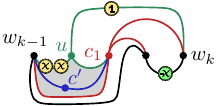}
  \end{minipage}\hfill
  \begin{minipage}[b]{.49\textwidth}
    \centering\includegraphics[page=31]{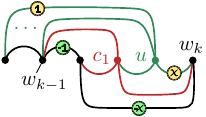}
  \end{minipage}\hfill
  \caption{Exactly one eligible vertex in~$X_{k-1}$ and it has right pivot
    type.\label{fig:onlyright}}
\end{figure}

It remains to consider the case that~$X_{k-2}$ is empty. Select~$j$ to be minimal such that all of~$X_{j+1},\ldots,X_{k-2}$ are empty. If~$\Delta\cap V=\emptyset$, then we place~$c_1$ into the pocket it covers and then place~$u$ into the pocket to the right of~$c_1$, see \cref{fig:onlyright}~(right). We place one credit on the mountain~$uw_{j+1}$ and~$\chi$ credits on the pocket~$uw_k$. As we can take~$1+\chi$ credits from the edges covered by~$c_1$, the costs are zero to establish \ref{i:uw} for~$X_j$, or end this step if~$j=0$. Otherwise, we have~$\Delta\cap V\ne\emptyset$ and proceed exactly as described above for the case that~$X_{k-2}$ is nonempty---except that we also include all of the regions~$X_{j+1},\ldots,X_{k-2}$ in the induction. To formally obtain a triangulation, we add virtual edges~$c_1w_{j+1},\ldots,c_1w_{k-2}$, which we immediately remove from the resulting drawing again. The analysis remains unchanged.

\subparagraph{X$_{\mathbf{k-1}}$ is a both pivot region.} 
By \cref{lem:rightpivot} there is exactly one vertex~$c_s$ of right pivot type
in~$X_j\cap\mathcal{E}_{i-1}$ and we have~$\mathrm{pc}(c_s)=u$. All vertices
in~$(X_j\cap V)\setminus\mathcal{E}_{i-1}$ (if any exist) are in the open
quadrilateral~$\Box=c_{s-1}w'c_su$.

If~$\Box\cap V=\emptyset$, then for each~$c_h$, with~$1\le h\le s$, we push down
the leftmost mountain covered by~$c_h$ and place~$c_h$ there. Then we place~$u$
into the pocket to the left of~$c_s$; see \cref{fig:bothpivot}~(left), which
shows the case~$s=2$. The costs are~$2-\chi\le 3(1-\chi)$, for~$\chi\le 1/2$,
which the at least three new vertices~$u,c_{s-1},c_s$ can pay. As~$w_{k-1}u$ is
a mountain that carries one credit, we can add more mountains from~$u$ to the
left in case there are empty regions there and then move the credit to the
leftmost such mountain. So we either establish \ref{i:uw} for some~$X_j$,
with~$1\le j\le k-2$, or we finish this step if~$X_{k-1}$ is the only nonempty
region.

\begin{figure}[htbp]
  \begin{minipage}[b]{.51\textwidth}
    \centering\includegraphics[page=30]{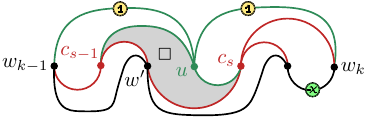}
  \end{minipage}\hfill
  \begin{minipage}[b]{.48\textwidth}
    \centering\includegraphics[page=40]{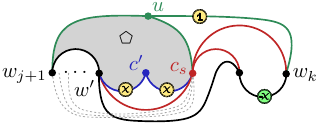}
  \end{minipage}
  \caption{There are eligible vertices of both pivot types
    in~$X_{k-1}$.\label{fig:bothpivot}}
\end{figure}

Otherwise, let~$j$ be minimal such that all of~$X_{j+1},\ldots,X_{k-2}$ are
empty. Note that we may have~$j=k-2$ if~$X_{k-2}$ is nonempty or~$j=0$
if~$X_{k-1}$ is the only nonempty region. Let~$c'$ be the vertex inside~$\Box$
that forms a triangle with~$w'c_s$ in~$G$, and note that~$w'$ is the only
neighbor of~$c'$ on~$P_\circ(G_{i-1})$. We place~$c_s$ by pushing down the
mountain it covers. Let~$\pentagon$ be the open region bounded the
path~$w'c'c_suw_{j+1}$ together with the part of~$P_\circ(G_{i-1})$
between~$w_{j+1}$ and~$w'$ in~$G$, and let~$G_{\pentagon}$ be the graph obtained
by adding the virtual edges~$w_{j+1}c_s,\ldots,w_{k-1}c_s$ (which are not
in~$G$) to the subgraph of~$G$ induced by the cycle~$\partial\pentagon$ together
with all vertices inside; see \cref{fig:bothpivot}~(right). Using
\cref{lem:main} we inductively obtain an extensible diagram~$D$
for~$G_{\pentagon}$, with~$w_{j+1}c_s$ as a starting edge, the profile $(w_{j+1},w',c',c_s)$ shown in
\cref{fig:bothpivot}~(right), and~$u$ as a last vertex. Then we remove the
virtual edges from~$D$, plug the resulting diagram into~$\pentagon$, add the
edge~$uw_k$ as a mountain, and place one credit on it. We also pay~$2\chi$
credits to initialize the two pockets incident to~$c'$ in~$D$ and
another~$2\chi$ credits for \cref{lem:main}. But we can take~$\chi$ credits from
the pocket covered by~$c_s$. So the costs to add~$c_s$ and~$c'$ are at
most~$1+3\chi\le 2(1-\chi)$, for~$\chi\le 1/5$, to either establish \ref{i:uw}
for~$X_j$ or finish this step.

\subsection{Processing the remaining regions}
If the initialization described in the previous section does not complete
processing of~$u$ and its regions already, then it establishes \ref{i:uw} for
some region~$X_j$, with~$1\le j\le k-1$. Denote the current working diagram
(for~$G[V_{i-1}\cup\bigcup_{h=j+1}^{k-1}X_h\cup\{u\}]$) by~$\Gamma$. As~$X_j$ is
nonempty by \ref{i:uw} the region~$X_j$ is either left, right, or both
pivot. These three different cases are discussed below. In all cases the
edge~$uw_j$ is drawn as a mountain and we place one credit on it. Therefore, any
number of empty regions~$X_h,\ldots,X_{j-1}$, for~$1\le h\le j$, are easy to
handle: Just add the edges~$uw_h,\ldots,uw_{j-1}$ and move the credit
from~$uw_j$ to~$uw_h$, to establish \ref{i:uw} for~$X_{h-1}$ or finish this step
if~$h=1$.

\subparagraph{X$_{\mathbf{j}}$ is a left pivot region.}
If~$s=1$ and~$\Delta\cap V=\emptyset$, then we place~$c_1$ by pushing down the
leftmost mountain it covers. We pay one credit for the mountain~$w_ju$, but we
can take the credits on~$w_{j+1}u$; see \cref{fig:allleftj}~(left). So by
\ref{i:uw} we pay at most~$1-2\chi$ credits, which the new vertex~$c_1$ is happy
to supply. If~$s=1$ and~$\Delta\cap V\ne\emptyset$, then we distinguish two
cases.

\begin{figure}[htbp]
  \begin{minipage}[b]{.49\textwidth}
    \centering\includegraphics[page=26]{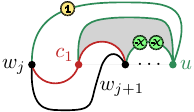}
  \end{minipage}\hfill
  \begin{minipage}[b]{.49\textwidth}
    \centering\includegraphics[page=25]{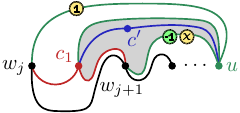}
  \end{minipage}
  \caption{Exactly one eligible vertex in~$X_j$, for~$j<k$, and it has left pivot
    type. Also observe that the solutions for subproblems reinserted into the gray shaded regions are actually rotated by $\pi$.}
    \label{fig:allleftj}
\end{figure}

If the edge~$w_{j+1}u$ is a mountain, then we push down the leftmost mountain
covered by~$c_1$. Using \cref{lem:mainstrong} we inductively obtain an
extensible diagram~$D$ for the subgraph~$G_\Delta$ of~$G$ induced by
taking~$\Delta$ as an outer triangle together with all vertices inside,
with~$uc_1$ as a starting edge and~$w_{j+1}$ as a last vertex. Note that~$D$
appears upside down compared to~$\Gamma$. But this~$180^\circ$ rotation is fine because
down-up biarcs remain down-up biarcs when turned upside down. Let~$c'$ be the
vertex that covers~$uc_1$ in~$D$.

The statement of \cref{lem:mainstrong} specifies two options. First we consider
the case that a valid diagram~$D$ for~$G_\Delta$ is obtained. In order to
plug~$D$ into~$\Delta$ in~$\Gamma$, we push down all mountains with left
endpoint~$u$ or~$w_{j+1}$ in~$D$ and all mountains with left endpoint~$w_{j+1}$
in~$\Gamma$; see \cref{fig:allleftj}~(right). By \ref{i:mountainMoney} for~$D$
and~$\Gamma$ these biarcs can be paid for using the corresponding mountain
credits. We pay~$2\chi$ to initialize the two pockets incident to~$c'$ in~$D$
and~$1-\chi$ for~$w_{j+1}$, which is part of~$\Gamma$ already. Further, as we
pay for~$uw_{j+1}$ as a part of~$D$, we get a refund for the~$1-\chi$ credits
that according to \ref{i:uw} are placed on~$uw_{j+1}$. We also need to place one
credit on the mountain~$w_ju$. So the costs to add~$c_1$ and~$c'$
are~$2\chi+(1-\chi)-(1-\chi)+1=1+2\chi\le 2(1-\chi)$, for~$\chi\le 1/4$.

The other option in \cref{lem:mainstrong} is that we obtain an extensible
diagram~$D'$ for~$G_\Delta\setminus\{w_{j+1}\}$ such that~$w_{j+1}$ is
problematic for~$D'$. Then we complete~$D'$ to a valid diagram~$D$
for~$G_\Delta$ as follows. As~$w_{j+1}$ has degree at least three in~$G_\Delta$,
there are three possibilities. If~$w_{j+1}$ is~$\mathcal{T}(3,\frown\smile)$
in~$D'$, then we just insert~$w_{j+1}$ into the pocket it covers, for a cost of
one (for the incident mountain). Then we proceed exactly as described above for
the first option of \cref{lem:mainstrong}, noting that we can save the~$\chi$
credits from the pocket incident to~$w_{j+1}$ in~$D$, so that in fact the costs
for inserting~$w_{j+1}$ into~$D'$ are~$1-\chi$, as they should be. Otherwise,
the vertex~$w_{j+1}$ is~$\mathcal{T}(z,\frown^{z-1})$ in~$D'$,
for~$z\in\{3,4\}$. We complete~$D'$ to~$D$ by placing~$w_{j+1}$ as the rightmost
vertex on the spine and drawing the edges incident to~$w_{j+1}$ as mountains, in
a similar fashion as for \cref{lem:mainadapt}. Next, we plug~$D$ into~$\Delta$,
making room by pushing down all mountains with left endpoint~$u$ in~$D$ and all
mountains with left endpoint~$w_{j+1}$ in~$\Gamma$; see
\cref{fig:allleftj2}~(left). As for the costs, we pay~$2\chi$ credits to
initialize the two pockets incident to~$c'$ in~$D$ and~$\chi$ credits for the
biarc~$uw_{j+1}$, which already carries~$1-\chi$ credits by \ref{i:uw}. We also
pay one credit each for the mountains~$w_jc_1$ and~$w_ju$, and we can take one
credit from one of the mountains covered by~$w_{j+1}$ in~$D$ that is not
incident to~$u$. So the costs to insert~$c_1$ and~$c'$ are at most~$2\chi+\chi+2-1=1+3\chi\le 2(1-\chi)$,
for~$\chi\le 1/5$.

\begin{figure}[htbp]
  \begin{minipage}[b]{.49\textwidth}
    \centering\includegraphics[page=35]{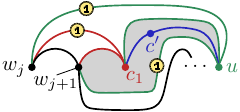}
  \end{minipage}\hfill
   \begin{minipage}[b]{.49\textwidth}
    \centering\includegraphics[page=24]{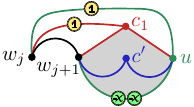}
  \end{minipage}
  \caption{Exactly one eligible vertex in~$X_j$, for~$j<k$, and it has left pivot
    type.\label{fig:allleftj2}}
\end{figure}

Otherwise, the edge~$w_{j+1}u$ is a pocket. Using \cref{lem:mainstrong} we
inductively obtain an extensible diagram~$D$ for the subgraph~$G_\Delta$ of~$G$
induced by taking~$\Delta$ as an outer triangle together with all vertices
inside, with~$w_{j+1}u$ as a starting edge and~$c_1$ as a last vertex; see
\cref{fig:allleftj2}~(right). Let~$c'$ be the vertex that covers~$w_{j+1}u$
in~$D$. Then we plug~$D$ into~$\Delta$. The~$2\chi$ credits to initialize the
two pockets incident to~$c'$ in~$D$ can be paid for by the credits that are
on~$w_{j+1}u$ by \ref{i:uw}. Then we need to pay one credit each for the
mountains~$w_jc_1$ and~$w_ju$. But in return we can take the at least~$1+\chi$
credits on the upper envelope of~$D$ because its edges are not incident to the
outer face anymore after adding the edge~$uw_j$. So by \cref{lem:mainstrong} the
costs to add~$c'$ are~$2-(1+\chi)=1-\chi$, or there exists an appropriate
extensible diagram~$D'$ for~$G_\Delta\setminus\{c_1\}$ for which~$c_1$ is
problematic. In the latter case, we just plug~$D'$ into~$\Gamma$ and then push
down~$w_jw_{j+1}$ to place~$c_1$ and draw all its edges to~$D'$ as
mountains. None of these edges remain on the outer face after adding~$w_ju$. So
the costs are one for~$w_ju$ and~$1\le 2(1-\chi)$, for~$\chi\le 1/2$, for the
two new vertices~$c_1$ and~$c'$.

It remains to consider the case~$s>1$. We place each vertex~$c_h$,
with~$1\le h<s$, by pushing down the leftmost mountain it covers, and
handle~$c_s$ in exactly the same way as in the case~$s=1$ described above. See
\cref{fig:manyleftslarge} for an example of the case
where~$\Delta\cap V\ne\emptyset$ and~$w_{j+1}u$ is a mountain. The edge~$w_jc_1$
is always a pocket, the only mountain that remains on the outer face is~$w_ju$,
and no additional biarc is created. Therefore, the same bounds on the costs as
for~$s=1$ also hold for~$s>1$, and in fact decrease by~$(s-1)(1-\chi)$, due to
the larger number of vertices added.

\begin{figure}[htbp]
    \centering\includegraphics[page=27]{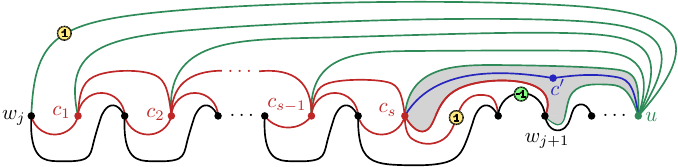}
    \caption{All $\ge 2$ eligible vertices in~$X_j$, for~$j<k$, have left pivot
      type.\label{fig:manyleftslarge}}
\end{figure}

\subparagraph{X$_{\mathbf{j}}$ is a right pivot region.}
By \cref{lem:rightpivot} there is exactly one vertex~$c_1$
in~$X_j\cap\mathcal{E}_{i-1}$ and~$\mathrm{pc}(c_1)=u$. Let~$\Delta$ denote the
open triangle~$w_kc_1u$. We push down the leftmost mountain covered by~$c_1$,
see \cref{fig:onlyright2}. If~$\Delta\cap V=\emptyset$, then the costs are at
most~$1-3\chi$, which the new vertex~$c_1$ can pay. Otherwise, we
have~$\Delta\cap V\ne\emptyset$ and using \cref{lem:mainadapt} we inductively
obtain an extensible diagram~$D$ for the subgraph of~$G$ induced by
taking~$\Delta$ as an outer triangle together with all vertices inside,
with~$w_jc_1$ as a starting edge and~$u$ as a last vertex. Then we plug~$D$
into~$\Delta$. Let~$c'$ be the vertex that covers~$w_jc_1$ in~$D$. The at
least~$2\chi$ credits that are on~$uw_{j+1}$ by \ref{i:uw} can pay the
initialization of the two pockets incident to~$c'$. We can also take the~$\chi$
credits from the pocket covered by~$c_1$. We pay~$1-\chi$ 
credits for~$u$ in~$D$. So the total costs
to insert~$c_1$ and~$c'$ are at most~$1-2\chi<2(1-\chi)$.

\begin{figure}[htbp]
  \begin{minipage}[b]{\textwidth}
    \centering\includegraphics[page=29]{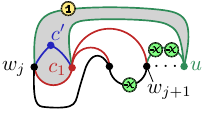}
  \end{minipage}
  \caption{All eligible vertices in~$X_j$, for~$j<k$, have right pivot
    type.\label{fig:onlyright2}}
\end{figure}

\subparagraph{X$_{\mathbf{j}}$ is a both pivot region.} 
By \cref{lem:rightpivot} there is exactly one vertex~$c_s$ of right pivot type
in~$X_j\cap\mathcal{E}_{i-1}$ and we have~$\mathrm{pc}(c_s)=u$. All vertices
in~$(X_j\cap V)\setminus\mathcal{E}_{i-1}$ are in the open
quadrilateral~$\Box=c_{s-1}w'c_su$. We argue in the same way as above in
\cref{sec:placement}, except that now~$u$ is placed to the right of~$w_k$ and so
we use \cref{lem:mainadapt} for the induction; see \cref{fig:bothpivot2}.

\begin{figure}[htbp]
  \begin{minipage}[b]{.49\textwidth}
    \centering\includegraphics[page=41]{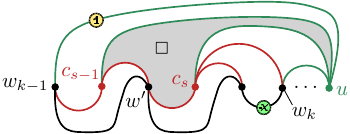}
  \end{minipage}\hfill
  \begin{minipage}[b]{.49\textwidth}
    \centering\includegraphics[page=42]{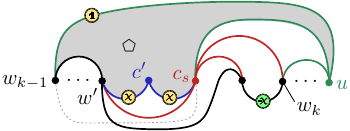}
  \end{minipage}
  \caption{There are eligible vertices of both pivot types in~$X_j$,
    for~$j<k$.\label{fig:bothpivot2}}
\end{figure}

\section{Omitted Proofs from \cref{sec:kleetopes}}
\label{app:kleetopes}

\thmdegreethree*

\begin{proof}
\label{p:degreethree}
  Let~$G$ be as in the statement, and let~$T$ denote the triangulation that results from removing all degree three vertices from~$G$. Then~$T$ has~$k=n-d$ vertices. We proceed in two steps. 
  
  In the first step, we obtain a monotone plane arc diagram for~$T$ with at most~$k-4=n-d-4$ biarcs, where every biarc is down-up and such that every triangle~$t$ in the diagram crosses the spine, that is, the interior of~$t$ intersects the spine in a line segment. In the second step, we place all degree three vertices of~$G\setminus T$ in the drawing, each vertex on the spine segment of the triangle in~$T$ that contains it in~$G$ and connect it to each of the three vertices of the triangle by a proper arc. As no biarcs are created in the second step, it suffices to argue how to obtain a diagram for~$T$ fulfilling Invariants \ref{i:biarcTypes}--\ref{i:biarcMoney} and \ref{i:contour} in the first step.

  We choose any canonical ordering~$w_1,\ldots,w_k$ for~$T$. Then we start off by drawing the edge~$w_1w_2$ as a pocket, into which we insert~$w_3$ and draw the edge~$w_1w_3$ as a pocket and the edge~$w_3w_2$ as a mountain, onto which we place one credit. It is easily verified that this diagram satisfies \ref{i:biarcTypes}--\ref{i:biarcMoney} and \ref{i:contour} and that each triangle in the diagram crosses the spine. Then we insert the vertices~$w_4,\ldots,w_k$ one by one while maintaining an arc diagram for~$T_i=T[\{v_1,\ldots,v_i\}]$ that satisfies \ref{i:biarcTypes}--\ref{i:biarcMoney} and \ref{i:contour}, for~$i=4,\ldots,k$, along with the property that all triangles cross the spine. Note that for each biarc both incident triangles cross the spine. Hence, pushing down a mountain maintains the spine crossing property. When inserting a new vertex~$v_i$ we distinguish two cases.

\begin{figure}[htbp]
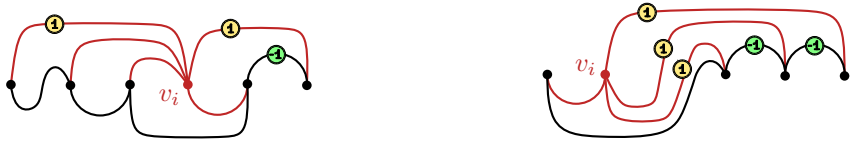

  \begin{minipage}[b]{.49\textwidth}
    \centering\includegraphics[page=44]{open}
  \end{minipage}\hfill
  \begin{minipage}[b]{.49\textwidth}
    \centering\includegraphics[page=43]{open}
  \end{minipage}
  \caption{Insert vertices so as to make every triangle cross the spine.\label{fig:klee:app}}
\end{figure}

  If~$v_i$ covers at least one pocket, then we place it into the rightmost pocket it covers; see \cref{fig:klee:app}~(left). All edges to vertices to the left of~$v_i$ are drawn as mountains. Only the mountain to the leftmost neighbor requires a credit so as to establish \ref{i:mountainMoney}. Every mountain covered by~$v_i$ to the left of~$v_i$ we push down, transforming it into a biarc. This transformation is paid for by the credit that is on the mountain by \ref{i:mountainMoney}. In this way we ensure that all triangles to the left of~$v_i$ cross the spine.
  The edge to the immediate neighbor of~$v_i$ to the right is drawn as a pocket. If there are further neighbors of~$v_i$ to the right, then we draw the edge to the rightmost neighbor as a mountain and all other edges in between as biarcs. The costs for each such mountain or biarc are paid using the credit on the mountain underneath whose left endpoint is covered by~$v_i$ (all underneath edges are mountains because we insert~$v_i$ into the rightmost pocket it covers). In this way we ensure that all triangles to the right of~$v_i$ cross the spine. Overall, the insertion of~$v_i$ costs one credit in this case.

  Otherwise, all edges covered by~$v_i$ are mountains; see \cref{fig:klee:app}~(right). We push down the leftmost such mountain to place~$v_i$ there. The edge to the immediate neighbor of~$v_i$ to the left is drawn as a pocket. The edge to the rightmost neighbor of~$v_i$ is drawn as a mountain, on which we place one credit. All other edges (which are to vertices in between that are to the right of~$v_i$) are drawn as biarcs. To pay for such a biarc we use the credit on the mountain underneath whose left endpoint is covered by~$v_i$. In this way we ensure that all triangles incident to~$v_i$ cross the spine. Overall, the insertion of~$v_i$ costs one credit in this case.

  It is easily checked that the Invariants \ref{i:biarcTypes}--\ref{i:biarcMoney} and \ref{i:contour} are maintained by the algorithm described above. Inserting each of~$v_3,\ldots,v_k$ costs one credit, which is~$k-2$ credits in total. Furthermore, we can (1)~take the credit spent to insert~$v_k$ because this mountain remains a mountain in the final diagram and (2)~observe that the edge on the outer face that is incident to~$v_2$ is a mountain in the diagram for~$T_j$, for all~$3\le j\le k$. In particular, as~$v_k$ has degree at least three in~$T$, its insertion covers the mountain on the outer face of the diagram for~$T_{k-1}$, and so we can also take back the credit on this mountain. Therefore, no more than~$k-4$ credits are spent in total. As by \ref{i:biarcMoney} every biarc in the diagram corresponds to a credit, the theorem follows.
\end{proof}

\end{document}